\documentclass[longbibliography]{amsart}
\usepackage[unicode=true,pdfusetitle, bookmarks=true,bookmarksnumbered=false,bookmarksopen=false, breaklinks=false,pdfborder={0 0 0},backref=false,colorlinks=false]{hyperref}
\hypersetup{colorlinks,linkcolor=myurlcolor,citecolor=myurlcolor,urlcolor=myurlcolor}
\usepackage{graphics,epstopdf,graphicx,amsthm,amsmath,amssymb,mathptmx,braket,colortbl,color,bm,framed,mathrsfs}
\usepackage[up]{subfigure}
\usepackage{tikz}
\usepackage{multirow}
\usepackage{enumerate}
\usepackage{float}
\usepackage{cite}

\usepackage{arydshln}

\usepackage{tabularx}
\usepackage{stmaryrd}

\usepackage{booktabs} 
\let\hbar\undefined
\usepackage{fouriernc} 

\makeatletter
\def\l@section{\@tocline{1}{0pt}{1pc}{}{}}
\def\l@subsection{\@tocline{2}{0pt}{1pc}{4.6em}{}}
\def\l@subsubsection{\@tocline{3}{0pt}{1pc}{7.6em}{}}
\renewcommand{\tocsection}[3]{%
  \indentlabel{\@ifnotempty{#2}{\makebox[2.3em][l]{%
    \ignorespaces#1 #2.\hfill}}}#3}
\renewcommand{\tocsubsection}[3]{%
  \indentlabel{\@ifnotempty{#2}{\hspace*{2.3em}\makebox[2.3em][l]{%
    \ignorespaces#1 #2.\hfill}}}#3}
\renewcommand{\tocsubsubsection}[3]{%
  \indentlabel{\@ifnotempty{#2}{\hspace*{4.6em}\makebox[3em][l]{%
    \ignorespaces#1 #2.\hfill}}}#3}
\makeatother
\definecolor{myurlcolor}{rgb}{0,0,0.9}

\newcommand{\proj}[1]{| #1\rangle\!\langle #1 |}

\newcommand{\ep}[1]{\langle #1 \rangle}

\DeclareMathOperator{\trace}{Tr}

\theoremstyle{plain}
\newtheorem{thm}{Theorem}
\newtheorem{lem}[thm]{Lemma}
\newtheorem{prop}[thm]{Proposition}
\newtheorem{cor}[thm]{Corollary}

\newtheorem{Def}[thm]{Definition}
\newtheorem{Rem}[thm]{Remark}

\newtheorem{Examp}[thm]{Example}

\newcommand*{\myproofname}{Proof}

\def\ot{\otimes}
\def\complex{\mathbb{C}}
\def\CN{\mathbb{C}}

\def\Z{\mathbb{Z}}

\def\cG{\mathcal{G}}

\def\cC{\mathcal{C}}

\newcommand{\CCC}{\mathcal C}

\newcommand{\CEE}{\mathcal E}

\newcommand{\BFF}{\mathbb F}

\newcommand{\CGG}{\mathcal G}

\newcommand{\CNN}{\mathcal N}

\newcommand{\CSS}{\mathcal S}

\newcommand{\be}{\begin{equation}}
\newcommand{\ee}{\end{equation}}

\renewcommand{\ge}{\geqslant}
\renewcommand{\geq}{\geqslant}
\renewcommand{\leq}{\leqslant}
\renewcommand{\le}{\leqslant}

\newcommand{\rank}{\operatorname{rank}}
\newcommand{\row}{\operatorname{row}}

\DeclareMathAlphabet{\mathcal}{OMS}{cmsy}{m}{n}

\makeatother

 \title{ Quantum locally recoverable code with intersecting recovery sets}

\author{Kaifeng Bu}
\author{Weichen Gu}
\author{Xiang Li}

\address{Department of Mathematics, Ohio State University, Columbus, OH 43210, USA}

\begin{document}

\maketitle

\begin{abstract}
We introduce the concept of quantum locally recoverable codes (qLRCs) with intersecting recovery sets. We derive a singleton-like bound for these codes by leveraging the additional information provided by the intersecting recovery sets. Furthermore, we provide a construction for qLRCs with intersecting recovery sets by introducing a variation of the hypergraph product. Finally, we apply our qLRC methods to obtain improved results for classical LRCs. These results may provide new insights into  the locality of
quantum error correction code.

\end{abstract}

\maketitle

\setcounter{tocdepth}{2}
\tableofcontents

\section{Introduction}

In classical computation, locally recoverable codes (LRCs) have become increasingly important in modern data storage systems due to their efficiency in handling errors and failures, particularly in distributed storage environments. A key feature of these codes is their ability to recover from any single error by accessing only a small, localized subset of the stored data~\cite{PD2014,TB14,TBIEEE2014,CM15,WZ2014}. This localized recovery significantly reduces the bandwidth and latency associated with the recovery process, a critical advantage for large-scale systems such as cloud storage and data centers. 

Quantum error correction codes (QECCs), analogous to their classical counterparts, have been developed to protect quantum information from errors arising from decoherence and other forms of quantum noise. In recent years, the locality properties of these codes have garnered significant attention, leading to the development of quantum low-density parity-check codes (qLPDCs) and quantum locally testable codes (qLTCs).

For example, qLDPCs are designed to correct errors in quantum data by using a small number of parity checks, which reduces the complexity of error correction procedures and makes the implementation more feasible with current technology~\cite{Panteleev2022asymptotically,Leverrier2022quantum,BEPRQ21,BEIEEE21,PKIEEE22,HHO21}.
Furthermore, qLTCs are quantum codes that can determine, by acting on only a small, ideally constant-sized, subset of qubits, whether a given state is a valid codeword~\cite{AE15,EH2017,Cross2024quantum,Wills2024npj,Wills25}. They are closely related to quantum checkable probabilistic proof (qPCP) conjecture~\cite{aharonov2013guest}.

Beyond these two types of quantum codes with locality constraints, Golowich and Guruswami introduced a quantum version of locally recoverable codes (qLRCs)~\cite{GG23}. These codes enable the recovery of a single-qudit error within a codeword by accessing only a small set of local qudits. This approach offers a new avenue for investigating locality properties in quantum codes.

Further research has expanded on the concept of qLRCs, exploring variations such as qLRCs within the CSS framework~\cite{Luo2023bounds}, constructions based on suitable polynomials from classical coding theory~\cite{Sharma2024quantum}, and generalizations addressing simultaneous failures of multiple devices~\cite{Galindo2024quantum}.

In this work, we investigate the qLRCs with multiple recovery sets, which means that there is more than one local subset of qudits capable of recovering an erased qudit.
Employing multiple recovery sets has the potential to enhance the performance and properties of qLRCs. Furthermore, these recovery sets are not required to be disjoint; indeed, Theorem 74 in~\cite{GG23} demonstrates that qLRCs with disjoint recovery sets are trivial. Therefore, we focus on qLRCs with intersecting recovery sets.

\subsection{Summary of main results}
We introduce the concept of qLRCs with intersecting recovery sets, formally defined in Definition \ref{241130def1}. We denote these codes as $(r,t,s)$-qLRC, where 
each qudit has $t$ recovery sets, the size of each recovery set is bounded by $r$, and the size of the intersection of any two recovery sets
for the same qudit is bounded by $s$. We derive a singleton-like bound for these codes and demonstrate its improvement over the previous bound presented in Ref. \cite{GG23}, leveraging the additional information provided by the intersecting recovery sets. Furthermore, we propose a construction for qLRCs with intersecting recovery sets by introducing a variation of the hypergraph product. Finally, we apply our qLRC methods to obtain improved results for classical LRCs.

\section{Preliminary}
In this section, let us introduce some basic knowledge about QECCs and qLRCs.
Here, we consider a quantum system with Hilbert space $\mathcal{H}$, which could represent a qubit system
$\complex^2$ or  a qudit system $\complex^d$. 
The $n$-fold tensor product of this system, denoted $\mathcal{H}^{\ot n}$, represents an $n$-qubit or $n$-qudit system.

\begin{Def}
    A quantum code $\mathcal{C}$ with parameters $[n,k,d]$ encodes $k$ logical qubits into $n$ physical qubits, defining 
    a $2^k$-dimensional subspace within the n-qubit Hilbert space. This code is 
    capable of correcting errors acting on $\lfloor (d-1)/2\rfloor$ qubits.
    \end{Def}

A key result in QECCs is the quantum Singleton bound, which establishes a fundamental trade-off between the number of physical qubits used for encoding, the number of encoded logical qubits, and the code's error-correction capability. This bound sets fundamental limits on code performance and guides code design.

\begin{lem}[Quantum Singleton bound]
Any QECC with parameters $[n,k,d]$ satisfies the following bound:
    \begin{align}
        k\leq n-2(d-1).
    \end{align}

\end{lem}

One important family of QECCs is called stabilizer codes. 
To define these codes, let us first introduce the Pauli operators.
The 1-qubit Pauli operators are generated by the Pauli X and Z operators, defined as:
\begin{align}
    &X:\ket{0}\to \ket{1}, \ket{1}\to \ket{0};\\
    &Z:\ket{0}\to \ket{0}, \ket{1}\to -\ket{1}.
\end{align}
The $n$-qubit Pauli group, denoted $\mathcal{P}_n$, is then generated by:
\begin{align}
    \mathcal{P}_n=\set{\pm 1, \pm i}\times \set{X^{\vec a}Z^{\vec b}: \vec a=(a_1, a_2,...,a_n), \vec b=(b_1, b_2, ...,b_n)\in \set{0,1}^{n}},
\end{align}
where $X^{\vec a}:=X^{a_1}\ot X^{a_2}\ot...\ot X^{a_n}$, and $Z^{\vec b}:=Z^{b_1}\ot Z^{b_2}\ot...\ot Z^{b_n}$.

A stabilizer group $G$ is an abelian subgroup of the Pauli 
group $\mathcal{P}_n$ (i.e., all elements commute) that does not include 
$-I$. Formally, this group  $G$ can be generated 
by $k$ Pauli operators that commute with each other, expressed as $G=\langle g_1,..., g_k\rangle$, where each 
$g_i\in \mathcal{P}_n$ and the commutator  $[g_i,g_j]=0, \forall i,j\in [k]$.

\begin{Def}
A quantum code $\mathcal{C}$ is a stabilizer code  if  its code space is the joint +1 
 eigenspace of some stabilizer group $G$. 
Formally, 
\begin{align}
    \mathcal{C}=\set{\ket{\psi}: g\ket{\psi}=\ket{\psi}, \forall g\in G}.
\end{align}
\end{Def}

An important class of stabilizer code is the Calderbank-Shor-Steane (CSS) code~\cite{CS96,Ste96}. In CSS codes,  the generators of the stabilizer group consist of either X-type (i.e., $X^{\vec a}$) or
Z-type (i.e., $Z^{\vec b}$) Pauli operators.
The structure of the stabilizer generator matrix for CSS codes is characterized as:

\begin{equation}
 \left( \begin{array}{cc}
    0   & H_Z \\
     H_X  & 0 
  \end{array}
  \right),
\end{equation}
where $H_XH^T_Z=0$.  
This constraint ensures that the rows of $H_X$ and $H_Z$ in the generator matrices, which correspond, respectively, to X-type and Z-type Pauli operators, are mutually orthogonal, guaranteeing commutation among all elements of the stabilizer group.
We will denote this stabilizer code as $CSS(H_X,H_Z)$.
Note that the above concepts, including Pauli operators, stabilizer codes, and CSS codes, can be easily generalized to qudit systems.

\begin{Def}[\cite{GG23}]
A qLRC with parameter $[n,k,d,r]$ encodes $k$ logical qubits/qudits into $n$ physical qubits/qudits
has a code distance $d$, and is equipped with a set of local recovery channels   $\mathrm{Rec}_j$ for each $j\in [n] $, satisfying the following properties: 

1. For each $j\in [n]$, there  exists a set $I_j\subseteq [n]$  of size $|I_j |\le r+1$ with $j\in I_j$, such that $\mathrm{Rec}_j$ is a quantum channel supported on qudits in $I_j $.

2. For every code state $\psi \in \CCC $ we have
$$\mathrm{Rec}_j  \left(\trace _j(\proj{\psi} )\right) =\proj{\psi}. 
$$

\end{Def}

The parameters in qLRC satisfy the following bound  \cite{GG23} :
\begin{align}\label{241226eq1}
k \le n- 2(d - 1) - \left\lfloor \frac{n-(d-1)}{r+1} \right\rfloor -  \left\lfloor \frac{n-2(d-1) -  \left\lfloor \frac{n-(d-1)}{r+1} \right\rfloor}{r+1} \right\rfloor,
\end{align}
which is tighter than the quantum Singleton bound
$k\le n- 2(d-1).$

\section{Quantum locally recoverable code with intersecting recovery sets}
In this section, we consider qLRCs with intersecting recovery sets and provide the following formal definition.
Let us first recall the definitions of classic $(r,t,s)$-LRCs \cite{kruglik2017one}.
Recall a classic $(r,t,s)$-LRC is a code $\CCC$ satisfying: for every $j\in [n]$, there exist $t$ subsets $\Gamma_1(j),...,\Gamma_t(j) \subseteq [n]$ such that,
(a)  $\forall j\in \Gamma_l(j)$ and $|\Gamma_l(j)|= r+1$, $\forall l \in\set{1,2,...,t}$;
(b) for every two code words $c_1,c_2\in \CCC$, if $c_1(j) \neq c_2(j)$, then
$c_1|_{\Gamma_l(j) \setminus \{j\} } \neq c_2|_{\Gamma_l(j) \setminus \{j\} }$, for $l\in \set{1,2,...,t}$;
(c) for every different $l_1,l_2\in [t]$, 
$|\Gamma_{l_1}(j)\cap \Gamma_{l_2}(j)| \le s+1$.

Now, let us give the formal definition of $(r,t,s)$-qLRCs.

\subsection{qLRCs with intersecting recovery sets}
\begin{Def}\label{241130def1}
An $(r,t,s)$-qLRC is a quantum code $\CCC$ satisfying the following properties:

1. Each qudit $j\in [n] $ is assigned with $t$ subsets $\Gamma_1(j),...,\Gamma_t(j) \subseteq [n]$,
each $\Gamma_l(j)$ has size $|\Gamma_l(j)|\le r+1$ and $j\in \Gamma_l(j)$, $\forall l\in \set{1,2,\dots,t}$;

2. For each $j\in [n] $, there exist $t$ local recovery channels $\mathrm{Rec}_{1,j},...,\mathrm{Rec}_{t,j}$,
such that $\mathrm{Rec}_{i,j}$ is a quantum channel supported on qudits in $\Gamma_i(j)$;

3. For any code state $\psi \in \CCC $ and any qudit $j\in [n]$,  we have
$$\mathrm{Rec}_{i,j}  \left(\trace _j(\psi )\right) =\psi,\quad \forall 1\le i\le t; 
$$

4. For any different $l_1,l_2\in [t]$, 
\[ |\Gamma_{l_1}(j)\cap \Gamma_{l_2}(j)| \le s+1.\]

\end{Def}

Note that, we use $\psi$ to denote the density matrix $\proj{\psi}$ for simplicity. 
One important result in proving the quantum Singleton bound is the following lemma, see~\cite{Huber2020quantum,GHW_IEEE22,GG23}. 

\begin{lem}\label{241028lem2}
For any quantum code $\CCC$ encoding $k$ logical qubits/qudits in  $n$ physical qubits/qudits.
If there are two disjoint sets $A_1,A_2\subseteq [n]$ such that
$\CCC$ can  correct the errors support on $A_1$, and can also correct the errors support on $A_2$. 
Then 
\[k\le n- |A_1| - |A_2|.\]
\end{lem}
\begin{proof}
For self-containing, we provide the proof.
First, by the Knill-Laflamme conditions~\cite{KL_97}, the reduced density matrix $\rho_{A_1}$
is identical for all code states $\rho\in\CCC$. Similarly, the reduced density matrix  $\rho_{A_2}$ is also identical.

Let $W = [n]\setminus (A_1\cup A_2)$.
Let $R$ be an ancilla system with $k$ qubits/qudits, and 
define the maximal entangled state $\ket{\phi}$ between the 
the ancilla system $R$ and the code space $\CCC$ as 
follows
\begin{align}
\ket{\phi} = \frac{1}{\sqrt{d^k}} \sum_{x\in \Z_d^k} \ket{x}_R\otimes  \ket{\mathrm{Enc}(x)} ,
\end{align}  
where $\mathrm{Enc} $
 is the encoding map of the code space $\CCC$.

Let us denote $S(A)$ to be the von Neumann entropy of the reduced state 
$\rho_A$ in the code state $\psi$. 
Then
\begin{align}\label{241027eq1}
S(R) + S(A_1) = S(R A_1) = S(A_2W) \le S(A_2) + S(W),
\end{align}
where the first equality  follows from the fact that 
 all code states $\ket{\psi}\in\CCC$ have the
same reduced density matrix $\rho_{A_1}$, 
and the second equality follows from the fact that $\psi$ is pure, 
and the inequality follows from the subadditivity of entropy.
Similarly, we also have 
\begin{align}\label{241027eq2}
S(R) + S(A_2) = S(R A_2) = S(A_1W) \le S(A_1) + S(W).
\end{align}
Combining equations \eqref{241027eq1} and \eqref{241027eq2}, we have
\[k= S(R)  \le   S(W) \le n-|A_1|-|A_2|, \]
which completes the proof.

\end{proof}

The following theorem provides an upper bound of the local dimension $k$ involving the parameters related to local recoverability. 
To state it, we need the following notation:
\begin{align}\label{241208eq1}
p_2(r,t,s) = \sum_{1\le L \le t, \; L \equiv 1\bmod 2} \binom{t}{L} \frac{1}{\overline{N} (r,L,s)} -\sum_{1\le L \le t, \; L \equiv 0\bmod 2} \binom{t}{L} \frac{1}{\underline{N} (r,L,s)},
\end{align}
where 
\begin{align}
 &\underline{N}(r, L, s)  \triangleq   \frac{(2 r-(x-1) s)}{2} x+1,\\
&\overline{N}(r, L, s) \triangleq L r+1,
\end{align}
and  $x=\min \set{L,\lfloor r / s\rfloor+1}$ for any $L\in [t]$.
Here $p_2$ will be used to estimate the probability of certain event.

\begin{thm}\label{241208thm1}
If $\CCC$ is an $[n,k]$ quantum qudit/qubit $(r,t,s)$-LRC,
then
\[k\le   n- \left\lceil n p_2(r,t,s)\right\rceil - \left\lceil (n- \left\lceil n p_2(r,t,s)\right\rceil)  p_2(r,t,s)\right\rceil,\]
where $ p_2(r,t,s) $ is defined as in \eqref{241208eq1}. 
\end{thm}
\begin{proof}
By Lemma \ref{241028lem2}, we only need to construct two disjoint sets $A\subseteq [n]$ and $B\subseteq [n]$ such that 
$|A|=\left\lceil n p_2(r,t,s) \right\rceil$, $|B|=\left\lceil (n- \left\lceil n p_2(r,t,s)\right\rceil)  p_2(r,t,s)\right\rceil$, 
and the errors  supported on $A$ or $B$ can be corrected. 

To do this, we use the idea from ~\cite{kruglik2017one} to construct 
a sequence of numbers, $j_1,...,j_N\in [n]$, where  
\[N = \left\lceil n p_2(r,t,s) \right\rceil, \]
such that,
for every $1\le k\le N$ there is at least one $\alpha$ satisfying
\begin{align} 
\Gamma_\alpha(j_k) \cap \{j_{1},j_{2},...,j_{k-1}\} = \emptyset.
\end{align} 

First, it has been shown in ~\cite{kruglik2017one} that
if $L\in [t]$ and $x=\min \{L,\lfloor r / s\rfloor+1\}$,  then
\begin{align}\label{eq:key_lem}
    \underline{N}(r, L, s) \leq\left|\bigcup_{l=1}^L \Gamma_l(j) \right| \leq \overline{N}(r, L, s).
\end{align}

Let us choose uniformly a total ordering $\pi$ on the set $[n]$,  
and let $U $ be the subset of $[n]$ defined as follows: 
a number $j$ belongs to $U$ iff there is at least one $\alpha $ such that $\pi(i)>\pi(j)$ for every $i\in \Gamma_\alpha(j)\setminus\{j\}$.
Then the numbers in $U$ ordered in $\pi$ (from small to large) form a  sequence $j_1,...j_m$ satisfying the condition that, for every $1\le k\le m$, there is at least one $\alpha$ such that
\[ \Gamma_\alpha(j_k) \cap \{j_{1},j_{2},...,j_{k-1}\} = \emptyset. \]
Now we need to find a $U$ with a large size.

For a given $j\in [n]$, let $E_\alpha$ be the event that $\pi(i)>\pi(j)$ for every  $i\in \Gamma_\alpha(j)\setminus\{j\}$.
From \eqref{eq:key_lem},
for any subset $S\subseteq [t]$ the probability of  all the events $\{E_\alpha:  \; \alpha\in S\}$ occur simultaneously can be estimated as follows
\begin{align*}
\frac1{\overline{N}(r, |S|, s)} \le \mathrm{Pr}\left( \bigcap_{\alpha\in S}E_\alpha\right) \le \frac1{\underline{N}(r, |S|, s)}.
\end{align*}
Hence
\begin{align*}
\mathrm{Pr}(j\in U) =& \mathrm{Pr}\left( \bigcup_{\alpha=1}^t E_\alpha\right)\\
=& \sum_{l=1}^t (-1)^{l-1} \sum_{S\subseteq [t],\; |S|=l} \mathrm{Pr}\left( \bigcap_{\alpha\in S}E_\alpha\right) \\
\ge &  \sum_{1\le l \le t, \; l \equiv 1\bmod 2} \binom{t}{l} \frac{1}{\overline{N} (r,l,s)} -\sum_{1\le l \le t, \; l \equiv 0\bmod 2} \binom{t}{l} \frac{1}{\underline{N} (r,l,s)}\\
=& p_2(r,t,s) .
\end{align*}

Let $X = |U|$ be the size of $U$, and let $X_j$ be the indicator random variable for $j\in U$, i.e., $X_j=1$ iff $j\in U$, and otherwise $X_j=0$.
For each $j\in [n]$, 
\[\mathbb{E}(X_j) = \mathrm{Pr}(j\in U) \ge p_2(r,t,s). \]
Therefore
\[\mathbb{E}(X ) = \sum_{j=1}^n \mathbb{E}(X_j) \ge n p_2(r,t,s).\]
Hence there must be a specific ordering $\pi$ with
\[|U| \ge   N.\]
Now we pick $N$ numbers in this $U$, and order them in the ordering $\pi$ (from small to large) to form a sequence $j_1,...j_N$.
From the construction of $U$, it satisfies the condition that, for every $1\le k\le N$, there is an $\alpha_k$ such that
\[ \Gamma_{\alpha_k}(j_k) \cap \{j_{1},j_{2},...,j_{k-1}\} = \emptyset. \]

Now, let us show that any Pauli error support on the set
\[A = \{ j_1,...,j_M\},\]
can be recovered.
Assume the Pauli error is given by
\[E= E_{j_1}  E_{j_2}  \cdots  E_{j_M},\]
where each $E_{j_k}$ is a Pauli error acting on the $j_k$-th qudit.
Since $\mathrm{Rec}_{\alpha_M,j_M}$ is a quantum channel supported in $\Gamma_{\alpha_M}(j_M)$ and $j_1,...,j_{M-1} \not\in \Gamma_{\alpha_1}(j_M)$, 
for every codeword $\ket{\psi}$ we have
\begin{align*}
&\mathrm{Rec}_{\alpha_M,j_M} 
\left(  E_{j_1}  E_{j_2}  \cdots  E_{j_M}   \ket{\psi}\bra{\psi}  E_{j_M}^\dag \cdots E_{j_2}^\dag E_{j_1}^\dag \right) \\
=&      E_{j_1}  E_{j_2}  \cdots E_{j_{M-1}}  \cdot\mathrm{Rec}_{\alpha_M,j_M} 
\left( E_{j_M}   \ket{\psi}\bra{\psi}  E_{j_M}^\dag \right) E_{j_{M-1}}^\dag \cdots E_{j_2}^\dag E_{j_1}^\dag   \\
=&    E_{j_1}  E_{j_2}  \cdots   E_{j_{M-1}}  \ket{\psi}\bra{\psi}   E_{j_{M-1}}^\dag \cdots E_{j_2}^\dag E_{j_1}^\dag.
\end{align*}
Similarly, since $j_{1},...,j_{k-1} \not\in \Gamma_{\alpha_k}(j_k)$ for $1\le k\le M$, we have
\begin{align*}
 &  \mathrm{Rec}_{\alpha_1,j_1} 
 \circ \cdots \circ \mathrm{Rec}_{\alpha_M,j_M} 
 \left( E_{j_1}  E_{j_2}  \cdots  E_{j_M}   \ket{\psi}\bra{\psi}  E_{j_M}^\dag \cdots E_{j_2}^\dag E_{j_1}^\dag  \right)
 = \ket{\psi}\bra{\psi}.
\end{align*}

In addition,
we can also find a sequences of numbers, $i_1,...,i_M\in [n]\setminus A$, where  
\[M= \left\lceil (n-|A|) p_2(r,t,s) \right\rceil, \]
such that,
for every $1\le k\le M$ there is at least one $\alpha$ satisfying
\begin{align} 
\Gamma_\alpha(i_k) \cap \{i_{1},i_{2},...,i_{k-1}\} = \emptyset.
\end{align} 

Let us choose uniformly a total ordering $\pi$ on the set $[n]$,  
and we let $U $ be the subset of $[n]\setminus A$ defined as follows: 
a number$j\in [n]\setminus A$ belongs to $U$ iff there is at least one $\alpha $ such that $\pi(i)>\pi(j)$ for every $i\in \Gamma_\alpha(j)\setminus\{j\}$.
Similarly,
for any $j\in [n]\setminus A$,
\begin{align*}
\mathrm{Pr}(j\in U) \ge p_2(r,t,s) .
\end{align*}
Therefore
\[\mathbb{E}(|U| ) = \sum_{j\in [n]\setminus A}  \mathrm{Pr}(j\in U) \ge (n-|A|) p_2(r,t,s).\]
Hence there must be a specific ordering $\pi$ with
\[|U| \ge   M.\]
Now we pick $M$ numbers in this $U$, and order them in the ordering $\pi$ (from small to large) to form a sequence $i_1,...,i_M$.
From the construction of $U$, it satisfies the condition that, for every $1\le k\le N$, there is an $\alpha_k$ such that
\[ \Gamma_{\alpha_k}(i_k) \cap \{i_{1},i_{2},...,i_{k-1}\} = \emptyset. \]

Due to the similar reason, we can  show that any Pauli error support on the set
\[B = \{ i_1,...,i_M\},\]
can be recovered. Therefore, we finish the proof.

\end{proof}

In Figure \ref{fig1}, we compare the upper bounds on $k$ in Theorem \ref{241208thm1} and Equation \eqref{241226eq1}, where $r,t,s$ are fixed.
Note that although the Theorem \ref{241208thm1} does not involve $d$,  it often provides a tighter upper bound for $k$ 
compared to Equation \eqref{241226eq1} in many examples. 
\begin{figure}[htbp]
    \centering
    \begin{minipage}{0.49\textwidth}
        \centering
        \includegraphics[width=\textwidth]{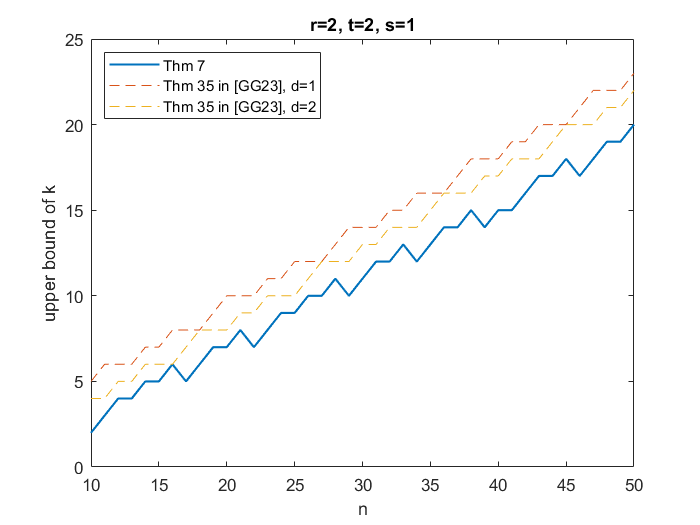}
    \end{minipage}
    \begin{minipage}{0.49\textwidth}
        \centering
        \includegraphics[width=\textwidth]{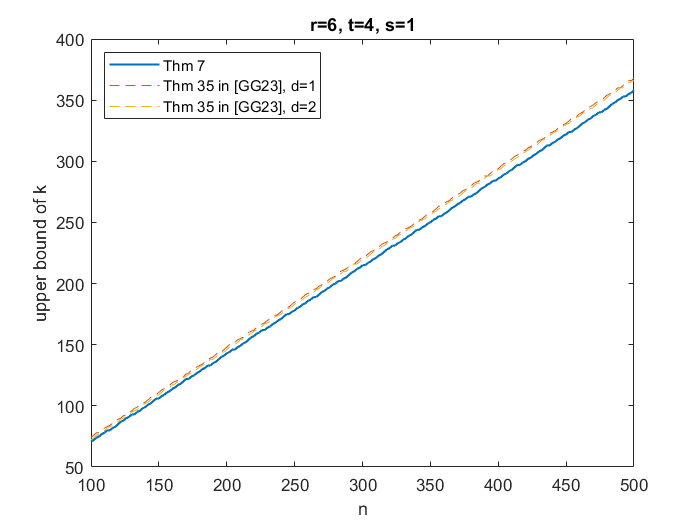}
    \end{minipage}
    \caption{Comparison of the upper bounds of $k$ from Theorem \ref{241208thm1} and Equation \eqref{241226eq1}.}
    \label{fig1}
\end{figure}

\subsection{qLRCs satisfying the exact $(r,t,s)$-condition}

In the previous section, we discussed qLRCs with intersecting recovery sets. We now consider a special case with additional constraints on the recovery sets,
which will make it easy to construct such codes.

\begin{Def}
Given $n,t,r,s \ge 1$.
A family $ \left(\Gamma_l(j)\right)_{1\le l\le t,\; 1\le j\le n}$ of subsets of $[n]$  is said to satisfy the exact $(r,t,s)$-condition if:

1. $|\Gamma_l(j)|= r+1$ and $j\in \Gamma_l(j)$ for all $j$ and $l$;

2. For any $j \in[n]$, $l_1,l_2\in [t]$ with $l_1\neq l_2$, 
\[ |\Gamma_{l_1}(j)\cap \Gamma_{l_2}(j)| = s+1;\]

3. For any distinct $l_1,l_2,l_3\in [t]$, 
\[\Gamma_{l_1}(j)\cap \Gamma_{l_2}(j) \cap \Gamma_{l_3}(j) = \{j\}.\]
\end{Def}

\begin{Def}\label{241130def2}
An exact $(r,t,s)$-qLRC is a quantum code $\CCC\subseteq (\CN^d )^{\otimes n}$ assigned with a family  $ \left(\Gamma_l(j)\right)_{1\le l\le t,\; 1\le j\le n}$ of subsets of $[n]$ satisfying exact $(r,t,s)$-condition, such that 
for each $j\in [n] $ there exist $t$ local recovery channels $\mathrm{Rec}_{1,j},...,\mathrm{Rec}_{t,j}$,
where each $\mathrm{Rec}_{l,j}$ is a quantum channel supported on qudits in $\Gamma_l(j)  $, and for every code state $\psi \in \CCC $ and every qudit $j\in [n]$ we have
$$\mathrm{Rec}_{l,j}  \left(\trace _j(\psi )\right) =\psi,\quad \forall 1\le l\le t. 
$$
\end{Def}


Similarly, we can also define the classic exact $(r,t,s)$-LRCs as follows. 
A classic exact $(r,t,s)$-LRC is a code $\CCC\subseteq \BFF_q^n$ assigned with a family  $ \left(\Gamma_l(j)\right)_{1\le l\le t,\; 1\le j\le n}$ of subsets of $[n]$ satisfying the exact $(r,t,s)$-condition, such that for every two code words $c_1,c_2\in \CCC$, if $c_1(j) \neq c_2(j)$, then
\[c_1|_{\Gamma_l(j) \setminus \{j\} } \neq c_2|_{\Gamma_l(j) \setminus \{j\} }, \quad l\in \set{1,2,...,t}.\]

Note that for quantum or classic exact $(r,t,s)$-LRCs, we always have
$r \ge s (t-1)$. 
In what follows, we investigate the parameters of these codes. First, we present the following lemma, which provides tighter control over the size of the union of the recovery sets.

\begin{lem}\label{241208lem6}
For any family $ \left(\Gamma_l(j)\right)_{1\le l\le t,\; 1\le j\le n}$ of subsets of $[n]$ which satisfies the exact $(r,t,s)$-condition, if $j\in [n]$ and $L\in [t]$,
then
\begin{align}\label{241209eq1}
    \left| \bigcup_{l=1}^L \Gamma_l(j) \right| = N_e (r,L,s)  \triangleq \frac{2r-(L-1)s}{2}L+1.   
\end{align} 
\end{lem}
\begin{proof}
When $L=1$, it is clear that  $\left|  \Gamma_1(j) \right| = 1+ r$. 
After the second set $\Gamma_2(j)$ is added, by exact $(r,t,s)$-condition, $(r-s)$ new elements are included and hence $\left| \bigcup_{l=1}^2 \Gamma_l(j) \right| = 1+ r+(r-s)$. 
Similarly, when the third set $\Gamma_3(j)$ is added, $(r-2s)$ new elements are included, so $\left| \bigcup_{l=1}^3 \Gamma_l(j) \right| = 1+ r+(r-s)+ (r-2s)$. In general, the size of the union can be expressed as follows
\[ \left| \bigcup_{l=1}^L \Gamma_l(j) \right| = 1+ r+ (r-s) + (r-2s) + \cdots + (r-(L-1)s) = 1+\frac{2r-(L-1)s}{2}L.\]

\end{proof}

We denote
\begin{align}\label{241208eq2}
p_e(r,t,s) = \sum_{1\le L\le t } (-1)^{L+1} \binom{t}{L} \frac{1}{N_e(r,L,s)}.
\end{align}
Then, for an $[n,k]$ exact $(r,t,s)$-qLRC, the following theorem establishes an upper bound for the parameter $k$.

\begin{thm}\label{241228thm2}
If $\CCC$ is an $[n,k]$ quantum qudit/qubit exact $(r,t,s)$-LRC,
then
\[k\le   n- \left\lceil n p_e(r,t,s)\right\rceil - \left\lceil (n- \left\lceil n p_e(r,t,s)\right\rceil)  p_e(r,t,s)\right\rceil,\]
where $ p_e(r,t,s) $ is defined as in \eqref{241208eq2}. 
\end{thm}
\begin{proof}
    The proof is similar to Theorem \ref{241208thm1}. 
    By Lemma~\ref{241028lem2}, we only need construct two disjoint sets $A\subseteq [n]$ and $B\subseteq [n]$ such that 
    $$|A|=\left\lceil n p_e(r,t,s)\right\rceil, \quad|B|=\left\lceil (n- \left\lceil n p_e(r,t,s)\right\rceil)  p_e(r,t,s)\right\rceil,$$
    and the errors supported on $A$ or $B$ can be corrected.

    To construct the sets $A$ and $B$, we only need to show that for any set
     $A\subseteq [n]$,
we can find a sequences of numbers, $j_1,...,j_N\in [n]\setminus A$, where  
\[N = \left\lceil (n-|A|) p_e(r,t,s) \right\rceil, \]
such that
for every $1\le k\le N$, there is at least one $\alpha$ satisfying
\begin{align} 
\Gamma_\alpha(j_k) \cap \{j_{1},j_{2},...,j_{k-1}\} = \emptyset.
\end{align} 

Let us choose uniformly a total ordering $\pi$ on the set $[n]$,  
and we let $U $ be the subset of $[n]\setminus A$ defined as follows: 
a number $j\in [n]\setminus A$ belongs to $U$ iff there is at least one $\alpha $ such that $\pi(i)>\pi(j)$ for every $i\in \Gamma_\alpha(j)\setminus\{j\}$.
Similar to  the proof of Theorem \ref{241208thm1},
for any $j\in [n]\setminus A$, by applying Lemma \ref{241208lem6}, we have
\begin{align*}
\mathrm{Pr}(j\in U) = p_e(r,t,s) .
\end{align*}
Therefore
\[\mathbb{E}(|U| ) = \sum_{j\in [n]\setminus A}  \mathrm{Pr}(j\in U) \ge (n-|A|) p_e(r,t,s).\]
Hence there must be a specific ordering $\pi$ with
\[|U| \ge   N.\]
Now we pick $N$ numbers in this $U$, and order them in the ordering $\pi$ (from small to large) to form a sequence $j_1,...j_N$.
This is the sequence needed. 

Therefore, we can get two sets $A=\set{j_1,...,j_N}$ and $B=\set{i_1,...,i_M}\subseteq [n]\setminus A$ with 
\begin{align}
  M=\left\lceil n p_e(r,t,s)\right\rceil, \quad N=\left\lceil (n- \left\lceil n p_e(r,t,s)\right\rceil)  p_e(r,t,s)\right\rceil,
\end{align}
 such that 
for every $1\le k\le M$ there is at least one $\alpha$ satisfying
\begin{align} 
\Gamma_\alpha(j_k) \cap \{j_{1},j_{2},...,j_{k-1}\} = \emptyset.
\end{align} 
and for every $1\le l\le N$ there is at least one $\alpha'$ satisfying
\begin{align} 
\Gamma_{\alpha'}(i_l) \cap \{i_{1},i_{2},...,i_{l-1}\} = \emptyset.
\end{align} 
Thus, by the similar reasons as in the proof of  Theorem~\ref{241208thm1}, the errors acting on the 
set $A$ or $B$ can be recovered.

\end{proof}

In the results above, we focused solely on the constraints imposed on the number of logical qubits/qudits 
$k$ and the number of physical qubits/qudits $n$ by the parameters related to local recoverability. Now, let us consider a quantum code with parameters 
 $[n,k,d]$  and local recoverability constraints, for which, we first 
 need the following lemmas. 

The first lemma concerns the size of a subgraph within a multigraph. Recall that a multigraph is a graph that allows multiple edges between the same pair of vertices (i.e., distinct edges sharing the same endpoints).

\begin{lem}\label{241227lem1}
Let $G= (V,E)$ be a  multigraph, and let $m\le |V|$.
 Then, there exists a subgraph $G'$ of $G$ containing $m$ vertices and at least 
\[\frac{m(m-1)}{|V|(|V|-1)} |E|.\]
edges.
\end{lem}
\begin{proof}	
Uniformly randomly choose a subset $V'\subseteq V$ of size $m$, and let $G'=(V',E')$ be the subgraph of $G$  induced by vertex set $V'$. 
For any edge $e\in E$, the probability that $e$ is included in $E'$ is 
\[\mathrm{Pr}(e\in E') = \frac{ \binom{|V|-2}{m-2}}{\binom{|V|}{m}}.\]
Thus, the expected number of edges in $E'$ is
\[\mathbb{E} |E'| = \sum_{e\in E}\mathrm{Pr}(e\in E') =   \frac{ \binom{|V|-2}{m-2}}{\binom{|V|}{m}}|E| = \frac{m(m-1)}{|V|(|V|-1)} |E|,\]
Therefore, there exists a choice of $G'$ containing at least $\frac{m(m-1)}{|V|(|V|-1)} |E|$ edges.
\end{proof}

\begin{lem}\label{241227lem2}
Let $A_1,...,A_M$ be $M$ subsets  of  $[n]$ with $|A_m|=r+1$ for $1\leq m\leq M$. Then,
for any  $N\leq M$, there exists a subset $S\subseteq [M]$ with $|S|=N$ such that 
\[\left| \bigcup_{m\in S} A_m\right| \le N(r+1) - \frac{N(N-1)}{M(M-1)}(M(r+1)-n) .\]

\end{lem}
\begin{proof}	
We construct an edge-colored graph $G=(V,E)$, where every edge is colored by one of the colors $c_1,...,c_n$.
Let $V=[M]$.
For $1\le j\le n$, let $T(j) = \{m\in [M] : j\in A_m  \}$ and assume  $m_j(1)<m_j(2)<\cdots < m_j(q_j) $ are all the elements in $T(j)$,
then we add the edges $$ \{m_j(1), m_j(2)\}, \; \{m_j(2), m_j(3)\} ,\;...,  \;\{m_j(q_j-1), m_j(q_j)\}$$ into $E$ and color them by $c_j$.
Here we complete the construction of $G$.
One property of $G$ is that, for every $j\in [n]$, the 
subgraph of $G$ consisting of edges colored with $c_j$ forms a tree, 
meaning it does not contain any cycles. 
(Figure \ref{fig:edge-colored-exm} is an example of an edge-colored graph.)
\begin{figure}[H]
    \centering
    \includegraphics[width=0.5\linewidth]{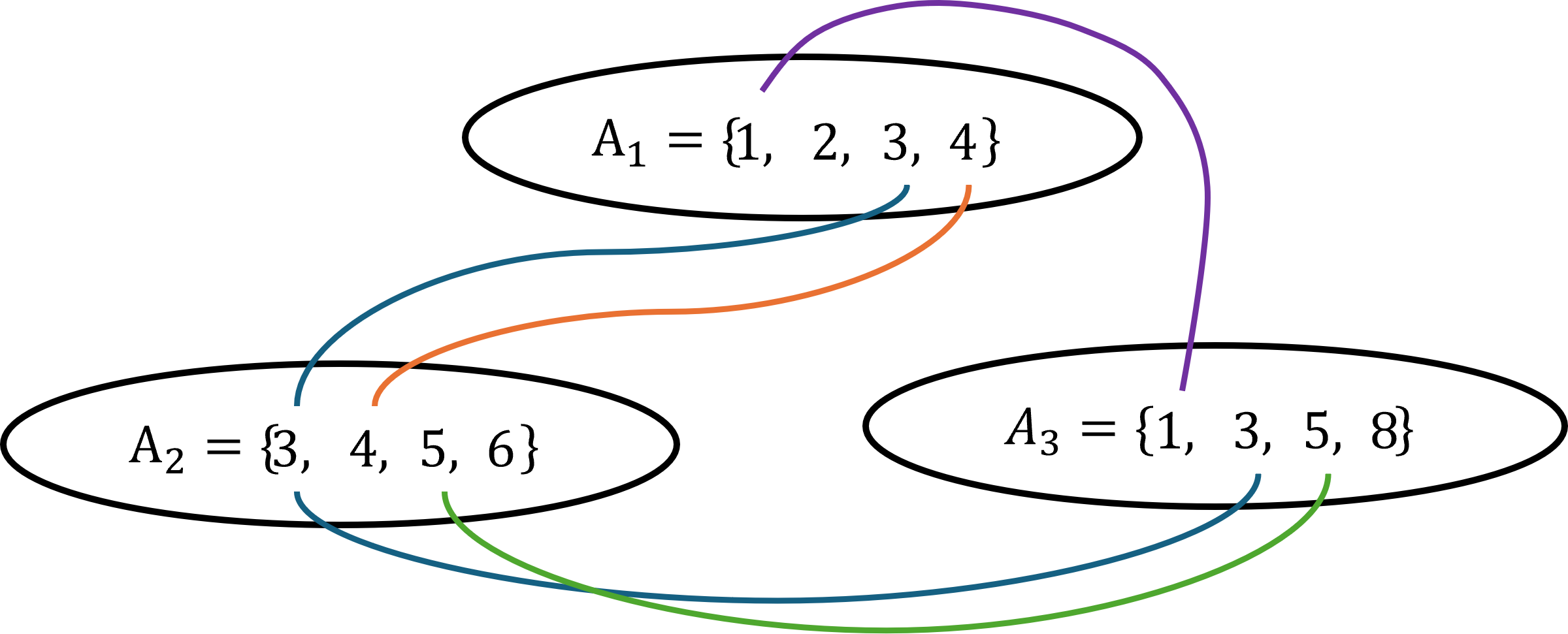}
    \caption{Example of an edge-colored graph $G(V,E)$, where $n=8$, $M=3$, and $r=3$. 
    The elements of the subsets $A_1, A_2, A_3$ are shown in the graph, and the vertice $i$ is equipped with the subset $A_{i+1}$ for $i=0,1,2$. 
    For $j = 1$, one purple edge is added; for $j = 3$, two blue edges are added; for $j = 4$, one orange edge is added; and when $j = 5$, one green edge is added.
}
    \label{fig:edge-colored-exm}
\end{figure}

Now, let us estimate the number of edges in graph $G$.
For every $j\in [n]$, assume that  $j$ appears $q_j$ times in the sets $A_1,...,A_M$. By construction, there are exactly $(q_j-1)$ edges in $G$ colored with $c_j$. 
Therefore, the total number of edges $|E|$ is given by:
\[|E| = \sum_{j\in [n]} (q_j-1) = \sum_{m\in [M]} |A_m| - \left| \bigcup_{m\in [M]} A_m \right| \ge M(r+1) -n,\]
where $|A_m|=r+1$, and $\left| \bigcup_{m\in [M]} A_m \right|\leq n$.

Moreover, by Lemma \ref{241227lem1}, we can find a subgraph $G'=(V',E')$ of $G$ induced by a subset $S\subseteq V = [M]$ with $|S|=N$ such that  
\[|E'| \ge \frac{N(N-1)}{M(M-1)} |E| \ge \frac{N(N-1)}{M(M-1)} (M(r+1) -n).\]
On the other hand,
for each $j\in [n]$, if $j$ appears $q_j'$ times in the sets $\{A_m: m\in S\}$, then there are at most $(q_j'-1)$ edges colored by $c_j$ in $G'$ as the  $c_j$-colored subgraph of $G'$ has no cycle.
Hence
\[ |E'| \le \sum_{j\in [n]} q_j'-1 =\sum_{m\in S} |A_m| - \left| \bigcup_{m\in S} A_m \right| = N(r+1) -\left| \bigcup_{m\in S} A_m \right|. \]
Therefore,
\[ \left| \bigcup_{m\in S} A_m \right| \le N(r+1) - \frac{N(N-1)}{M(M-1)}(M(r+1)-n). \]
\end{proof}

Now, let us define $N_1(n,r,d,M)$ as follows for simplicity:
If $n \ge \max\{r+1,d-1\}$, then 
\begin{align}
      N_1(n,r,d,M):=\max\left\{N : 0\le N\le M,\;   N(r+1) - \frac{N(N-1)}{M(M-1)}(M(r+1)-n) \le n-(d-1) \right\}.
\end{align}
If  $n< \max\{r+1,d-1\}$, then
\begin{align}
    N_1(n,r,d,M):=0.
\end{align}
Note that when $M \ge \frac{n}{r+1}$, we have
\[N_1(n,r,d,M) \ge \left\lfloor \frac{n-(d-1)}{r+1} \right\rfloor. \]

\begin{lem}\label{241228lem3}
Let $ \left(\Gamma_l(j)\right)_{1\le l\le t,\; 1\le j\le n}$ be a family   of subsets of $[n]$ satisfying the exact $(r,t,s)$-condition.
For any set $A\subseteq [n] $ with $n-|A| \ge \max\{r+1,d-1\}$ and $d-1<n$,
we can find a sequences of numbers, $j_1,...,j_M\in [n]\setminus A$, where  
\begin{align} 
M = (d-1) + N_1\left(n-|A|,r,d, \left\lceil (n-|A|) p_e(r,t,s) \right\rceil \right),
\end{align} 
such that
for every $d\le k\le M$ ,there is at least one $\alpha$ satisfying
\begin{align} 
\Gamma_\alpha(j_k) \cap \{j_{1},j_{2},...,j_{k-1}\} = \emptyset.
\end{align} 
\end{lem}
\begin{proof}
By the proof in Theorem \ref{241228thm2},
we can find a sequence of numbers
$i_1,...,i_{M'}\in [n]\setminus A$ with 
$$M' = \left\lceil (n-|A|) p_e(r,t,s) \right\rceil,$$
	such that,
	for every $1\le k\le M'$, there is an $\alpha_k$ satisfying
	\begin{align}\label{241023eq1}
		\Gamma_{\alpha_k}(i_k) \cap \{i_{1},i_{2},...,i_{k-1}\} = \emptyset.
	\end{align}
    
We apply Lemma \ref{241227lem2} on the sets $\{\Gamma_{\alpha_k}(i_k)\cap ([n]\setminus A): \;k=1,...,M'\}$, with 
$$N := N_1(n-|A|,r,d,M').$$
Since the set $\Gamma_{\alpha_k}(i_k)\cap ([n]\setminus A)$ may have size smaller than $r+1$,  for each $k$ we denote $A_k$ to be a set such that 
$$\Gamma_{\alpha_k}(i_k)\cap ([n]\setminus A) \subseteq A_k \subseteq [n]\setminus A, $$ and $|A_k|=r+1$.
Then by Lemma \ref{241227lem2},
we can find a subset $S\subseteq [M']$ with size $|S|=N$ such that
\[\left| \bigcup_{m\in S} A_m\right| \le N(r+1) - \frac{N(N-1)}{M'(M'-1)}(M'(r+1)-(n-|A|)) \le n-|A|-(d-1),\]
where the last inequality comes from the definition of $N_1(n-|A|,r,d,M')$.

Let $j_1,...,j_{d-1}$ be $(d-1)$ elements in $[n]\setminus \left(A \cup \bigcup_{k\in S} A_k\right)$,
then for every $k\in S$, 
$$\Gamma_{\alpha_k}(i_k) \cap \{j_{1},j_{2},...,j_{d-1}\} = \emptyset.$$
Moreover, let $ j_d,j_{d+1},..., j_M$ be all the elements in $S$ (ordered as in the sequence $i_1,...,i_{M'}$.)
Then for every $d\le k\le M$, there is at least one $\alpha$ satisfying
	\begin{align*} 
		\Gamma_\alpha(j_k) \cap \{j_{1},j_{2},...,j_{k-1}\} = \emptyset.
	\end{align*}
\end{proof}

\begin{thm}\label{241228thm4}
Suppose $n,k,d,r,t\ge 1$ are given.
If
$\CCC$ is an exact $(r,t,s)$-qLRC with parameters  $[n,k,d]$,
then
\[k\le  \max\{0, n-M_1-M_2\},\]
where 
\[M_1 =(d-1) + N_1(n,r,d, \left\lceil n p_e(r,t,s) \right\rceil) \]
and
\[M_2 =(d-1) + N_1(n-M_1,r,d, \left\lceil (n-M_1) p_e(r,t,s) \right\rceil). \]
\end{thm}
\begin{proof}
The proof follows a similar approach to that of Theorem \ref{241208thm1}, utilizing Lemma \ref{241028lem3} and Lemma \ref{241228lem3}.
\end{proof}

\begin{Rem}
    Note that the bound given by Theorem \ref{241228thm4} is at least as good as the bound \eqref{241226eq1}, for any $t\ge 1$.
To see this, we first show $p_e(r,t,s) \ge \frac{1}{r+1}$.
Recall $p_e(r,t,s)$ is a probability of the certain event in the proof of Theorem \ref{241228thm2}: fix a $j\in [n]$, for $1\le \alpha \le t$ we let $E_\alpha$ be the event that $\pi(i)>\pi(j)$ for every  $i\in \Gamma_\alpha(j)\setminus\{j\}$,
then 
\[p_e(r,t,s) =  \mathrm{Pr}\left( \bigcup_{\alpha=1}^t E_\alpha\right).\]
In particular, we have
\begin{align}
p_e(r,t,s) \ge  \mathrm{Pr}\left(   E_1\right)= \frac{1}{r+1}.
\end{align}
Hence
$$N_1(n,r,d, \left\lceil n p_e(r,t,s)  \right\rceil) \ge \left\lfloor \frac{n-(d-1)}{r+1} \right\rfloor,$$
and 
$$N_1(n-M_1,r,d, \left\lceil (n-M_1) p_e(r,t,s) \right\rceil) \ge \left\lfloor \frac{n-(d-1)-M_1}{r+1} \right\rfloor.$$
Therefore, the bound in Theorem \ref{241228thm4} is at least as strong as \eqref{241226eq1}.
\end{Rem}

Now, let us explore a different method to derive another bound for exact qLRCs. 
Define
\begin{align}\label{241208eq3}
\tilde{f_e}(n,d,r,t,s) = \sum_{1\le L\le t } (-1)^{L+1} \binom{t}{L} \frac{\binom{n-N_e(r, L, s) }{ d}}{\binom{n}{d}} \frac{1}{N_e(r,L,s)},
\end{align}
where $N_e(r,L,s)$ is defined in \eqref{241209eq1}.
Here we define $\binom{a}b =0$ if $a<b$, and $\tilde{f_e}(n,r,t,s)=0$ when $t=0$.

\begin{lem}\label{241208lem5}
Let $ \left(\Gamma_l(j)\right)_{1\le l\le t,\; 1\le j\le n}$ be a family   of subsets of $[n]$ satisfying the exact $(r,t,s)$-condition.
For  any set $A\subseteq [n] $ with $n-|A| \ge d-1$ and $d-1<n$,
if $n-|A| \ge N_e(r, t, x) $, then
we can find a sequences of numbers $j_1,...,j_N\in [n]\setminus A$, where  
\begin{align} 
N = (d-1) + 
  \left\lceil  (n-|A|) \tilde{f_e}\big(n-|A|,d-1,r,t,s\big)  \right\rceil,
\end{align} 
such that
for every $d\le k\le N$, there exists at least one $\alpha$ satisfying
\begin{align} 
\Gamma_\alpha(j_k) \cap \{j_{1},j_{2},...,j_{k-1}\} = \emptyset.
\end{align} 
\end{lem}
\begin{proof}

Denote $A^c=[n]\setminus A$.
Choose uniformly a total ordering $\pi$ in the set $A^c$,  
and let $U $ be the subset of $A^c$ defined as follows:
a number $j\in A^c$ belongs to $U$ if $\pi(j) < d$ or there exists at least one $\alpha $ such that $\pi(i)>\pi(j)$ for every $i\in \Gamma_\alpha(j)\setminus\{j\}$.
Then the numbers in $U$, ordered according to $\pi$ (from small to large), form a  sequence $j_1,...j_m \in A^c$ satisfying the condition: for every $d\le k\le m$, there exists at least one $\alpha$ such that
\[ \Gamma_\alpha(j_k) \cap \{j_{1},j_{2},...,j_{k-1}\} = \emptyset. \]

The next step is to find a $U$ with a large size.
For a given $j\in A^c$, 
let $\CEE$ be the event that $\pi(j) < d$, and
let $E_\alpha$ be the event that $\pi(i)>\pi(j)$ for every  $i\in (\Gamma_\alpha(j)\cap A^c)\setminus\{j\}$.
Hence, for every $j\in A^c$,
\begin{align*}
 \mathrm{Pr}(j\in U) 
=  \mathrm{Pr}\left( \CEE\cup \bigcup_{\alpha=1}^t E_\alpha\right) 
=  \mathrm{Pr}\left( \CEE \right)+ \mathrm{Pr}\left(  \CEE^c \cap \left(\bigcup_{\alpha=1}^t E_\alpha\right)\right) 
= \frac{d-1}{n-|A|}+ \mathrm{Pr}\left(  \CEE^c \cap \left(\bigcup_{\alpha=1}^t E_\alpha\right)\right).
\end{align*}

Now, let us  estimate $\mathrm{Pr}\left(  \CEE^c \cap \left(\bigcup_{\alpha=1}^t E_\alpha\right)\right)$.
Let us denote
\[ \Gamma_\alpha'(j)= \Gamma_\alpha (j)\cap A^c,\]
for each $\alpha$.
Given that  $j$ is fixed, and $N_e(r, t, x) = \left| \bigcup_\alpha \Gamma_\alpha (j) \right|$,
define
\[ \CSS := \bigcup_\alpha \Gamma_\alpha'(j) =  \left(\bigcup_\alpha \Gamma_\alpha (j) \right) \cap A^c, \]
and let $\delta = N_e(r, t, x) - |\CSS| $.
Define $V_2 = \left(\bigcup_\alpha \Gamma_\alpha (j) \right)\cap A$, hence $|V_2| = \delta$.
Since we assume
$n-|A|\ge N_e(r, t, x)$,
we can find a set $V_1\subseteq A^c \setminus \CSS$ such that  $|V_1| = \delta$.
Let $\phi$ be a permutation on $[n]$ such that:

1. $\phi(i) = i$ for $i\in [n]\setminus (V_1\cup V_2)$;

2. $\phi(V_1) = V_2$, $\phi(V_2) = V_1$, and  $\phi^2 = \mathrm{id}$. 

By the construction of $\phi$, we have:

1. $ \left(\phi\big(\Gamma_l(j)\big)\right)_{1\le l\le t,\; 1\le j\le n}$ forms a family   of subsets of $[n]$ that satisfies the exact $(r,t,s)$-condition with $\phi(j) = j$ for each $j$;

2. $\phi\big(\Gamma_\alpha(j)\big) \subseteq A^c$ for every $\alpha$;

3. $\phi $ acts as the identity w on $\Gamma_\alpha'(j)$, and $\Gamma_\alpha'(j)\subseteq \phi\big(\Gamma_\alpha(j)\big)$ for every $\alpha$.

Recall that $E_\alpha$ is the event that $\pi(i)>\pi(j)$ for every $i\in \Gamma'_\alpha(j)\setminus\{j\}$.
Let $F_\alpha$ be the event that $\pi(i)>\pi(j)$ for every $i\in \phi\big(\Gamma_\alpha(j)\big)\setminus\{j\}$,
then $ F_\alpha \subseteq E_\alpha$ for each $\alpha$.
Hence we have
$$\mathrm{Pr}\left(  \CEE^c \cap \left(\bigcup_{\alpha=1}^t E_\alpha\right)\right) \ge \mathrm{Pr}\left(  \CEE^c \cap \left(\bigcup_{\alpha=1}^t F_\alpha\right)\right).$$

 Next, we estimate $\mathrm{Pr}\left(  \CEE^c \cap \left(\bigcup_{\alpha=1}^t F_\alpha\right)\right)$.
According to Lemma \ref{241208lem6},
for any subset $S\subseteq [t]$, the probability of all the events  $F_\alpha$ for $\alpha\in S$ occur simultaneously is given by
\begin{align*}
\mathrm{Pr}\left( \bigcap_{\alpha\in S}F_\alpha\right) = \frac1{N_e(r, |S|, x)},
\end{align*}
and the probability of the events $\CEE^c$ and $F_\alpha$ for $\alpha\in S$ all occur simultaneously is
\begin{align*}
\mathrm{Pr}\left( \CEE^c\cap \bigcap_{\alpha\in S}F_\alpha\right)
= 
\frac{\binom{n-|A|-N_e(r, |S|, x) }{ d-1}}{\binom{n-|A|}{d-1}}  \frac1{N_e(r, |S|, x)},
\end{align*}
which comes from the fact the event $\CEE^c\cap \bigcap_{\alpha\in S} F_\alpha$ happens iff every element $i\in \bigcup_{\alpha\in S} \phi\big(\Gamma_\alpha(j)\big)$ has order $\pi(i)\ge d$ and $j$ is the smallest element in $\bigcup_{\alpha\in S} \phi\big(\Gamma_\alpha(j)\big)$ ordered in $\pi$.

Hence for every $j\in [n]\setminus A$,
\begin{align*}
&\mathrm{Pr}(j\in U)  
=  \mathrm{Pr}\left( \CEE \cup \bigcup_{\alpha=1}^t E_\alpha\right)\\
=& \mathrm{Pr}\left(\CEE \right)+ \mathrm{Pr}\left( \CEE^c \cap \left(\bigcup_{\alpha=1}^t E_\alpha\right)\right)\\
\ge & \mathrm{Pr}\left(\CEE \right)+ \mathrm{Pr}\left( \CEE^c \cap \left(\bigcup_{\alpha=1}^t F_\alpha\right)\right)\\
=& \frac{d-1}{n-|A|} 
+ \sum_{l=1}^t (-1)^{l-1} \sum_{S\subseteq [t],\; |S|=l} \mathrm{Pr}\left(\CEE^c \cap \bigcap_{\alpha\in S}F_\alpha\right) \\
= & \frac{d-1}{n-|A|}+  
 \sum_{1\le l \le t }(-1)^{l-1} \binom{t}{l}   
\frac{\binom{n-|A|-N_e(r, l, x) }{ d-1}}{\binom{n-|A|}{d-1}}  \frac1{N_e(r, l, x)}  \\
=& \frac{d-1}{n-|A|} + \tilde{f_e}(n-|A|,d-1,r,t,s) .
\end{align*}

Let $X = |U|$ be the size of $U$, and let $X_j$ be the indicator random variable for $j\in U$, where $X_j=1$ if $j\in U$, and  $X_j=0$ otherwise.
For each $j\in [n]\setminus A$, 
\[\mathbb{E}(X_j) = \mathrm{Pr}(j\in U) = \frac{d-1}{n-|A|} + \tilde{f_e}(n-|A|,d-1,r,t,s). \]
Thus, the expected size of $U$,
\[\mathbb{E}(X ) = \sum_{j\in [n]\setminus A}  \mathbb{E}(X_j) = (d-1)+ (n-|A|) \tilde{f_e}(n-|A|,d-1,r,t,s).\]
Therefore, there exists a specific ordering $\pi$ with
$|U| \ge   N$.
\end{proof}

\begin{cor}\label{250101cor1}
Let $ \left(\Gamma_l(j)\right)_{1\le l\le t,\; 1\le j\le n}$ be a family   of subsets of $[n]$ satisfying the exact $(r,t,s)$-condition.
Given a number $1\le d<n$, and any set $A\subseteq [n] $ with $n-|A| \ge d$.
Then we can find a sequences of numbers, $j_1,...,j_N\in [n]\setminus A$, where $$N=(d-1)+ \CNN_1(n-|A|,d-1,r,t,s ),$$ and 
\begin{align}\label{241209eq2}
\CNN_1(m,d-1,r,t,s )\triangleq \max_{0\le T \le t   \text{ with }  N_e(r, T, s) \le m} 
  \left\lceil  m \tilde{f_e}\big(m,d-1,r,T,s\big)  \right\rceil,
\end{align} 
such that,
for every $d\le k\le N$ there is at least one $\alpha$ satisfying
\begin{align} 
\Gamma_\alpha(j_k) \cap \{j_{1},j_{2},...,j_{k-1}\} = \emptyset.
\end{align} 
\end{cor}
\begin{proof}
Note that when $T\le t$ in \eqref{241209eq2}, a family of subsets of $[n]$ satisfying the exact $(r,t,s)$-condition contains a family of subsets of $[n]$ satisfying the exact $(r,t,s)$-condition. 
Hence the statement comes directly from Lemma \ref{241208lem5}.
\end{proof}

\begin{thm}\label{thm:bnd-rts-2}
If $\CCC$ is an $[n,k,d]$ quantum qudit/qubit exact $(r,t,s)$-LRC,
then
\[k\le   n-2(d-1)- \CNN_1\big(n ,d-1,r,t,s \big)-\CNN_1\big(n-(d-1)-\CNN_1(n ,d-1,r,t,s ),d-1,r,t,s \big),\]
where $\CNN_1$ is defined as in \eqref{241209eq2}.

\end{thm}
\begin{proof}
The proof is similar to that of  Theorem \ref{241208thm1} by using Lemma \ref{241028lem3},
and Corollary \ref{250101cor1}.

\end{proof}

In Figure \ref{fig:comp-22-27}, we compare the upper bounds in Theorem \ref{241228thm4} and Theorem \ref{thm:bnd-rts-2}, where $d,r,t,s$ are fixed.
We can observe that depending on the specific values of the parameters, one theorem may offer a tighter bound than the other. 
\begin{figure}[htbp]
    \centering
    \begin{minipage}{0.48\textwidth}
        \centering
        \includegraphics[width=\textwidth]{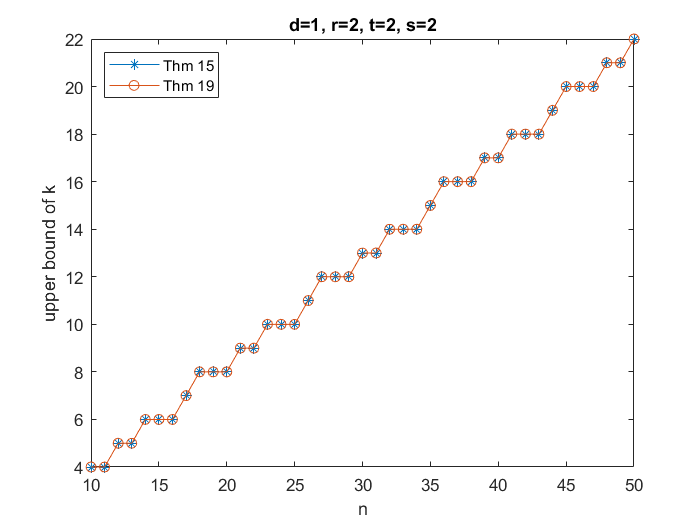}
    \end{minipage}\hfill
    \begin{minipage}{0.48\textwidth}
        \centering
        \includegraphics[width=\textwidth]{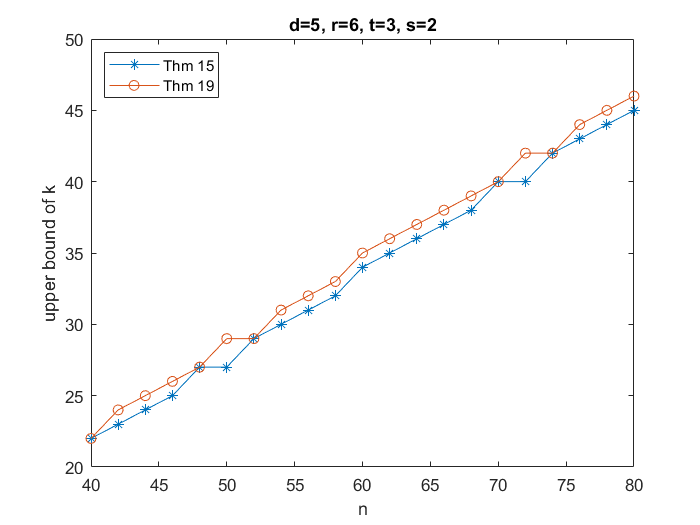}
    \end{minipage}
    
    \vskip\baselineskip

    \begin{minipage}{0.48\textwidth}
        \centering
        \includegraphics[width=\textwidth]{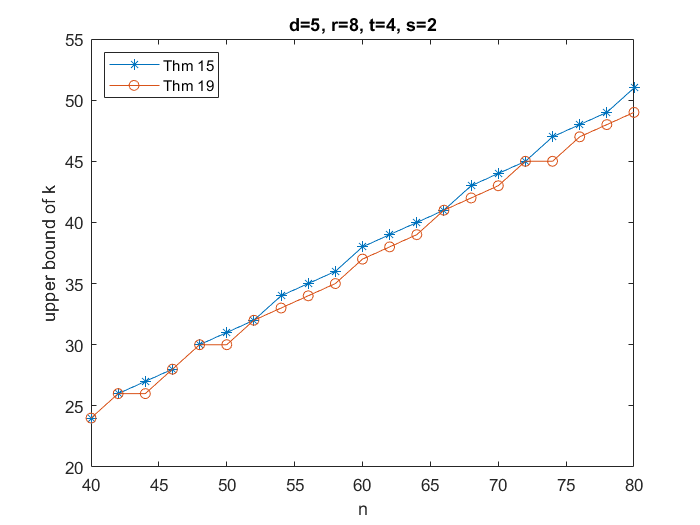}
    \end{minipage}\hfill
    \begin{minipage}{0.48\textwidth}
        \centering
        \includegraphics[width=\textwidth]{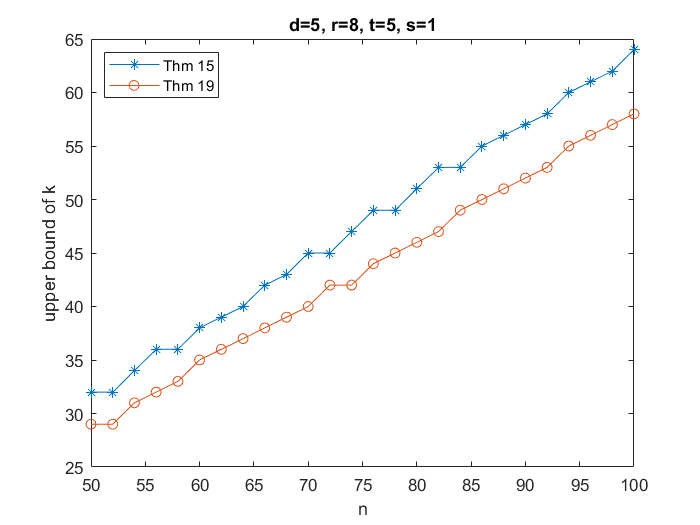}
    \end{minipage}
    
    \caption{Comparison of the upper bounds of $k$ from Theorem \ref{241228thm4} and Theorem \ref{thm:bnd-rts-2} }
    \label{fig:comp-22-27}
\end{figure}

\section{Construction of qLRCs}

In this section, we provide a method to construct exact qLRCs, based on 
the graph product of Tanner graphs. 

Recall that any parity-check matrix $H$ can be represented by a bipartite graph $(V,W,E)$, also referred to as its Tanner graph, which is constructed as follows:
vertices in $V$ (called bit nodes) correspond to columns of $H$,
and vertices in $W$ (called check nodes) correspond to rows of $H$;
a check node and a bit node are connected by an edge in $E$ if and only if there is a nonzero entry at the intersection of the corresponding row and column of $H$.
(See Fig. \ref{fig:tanner} for an example of the Tanner graph and the parity-check matrix. )
\begin{figure}[H]\label{fig:tanner}
    \centering
    \includegraphics[width=0.6\linewidth]{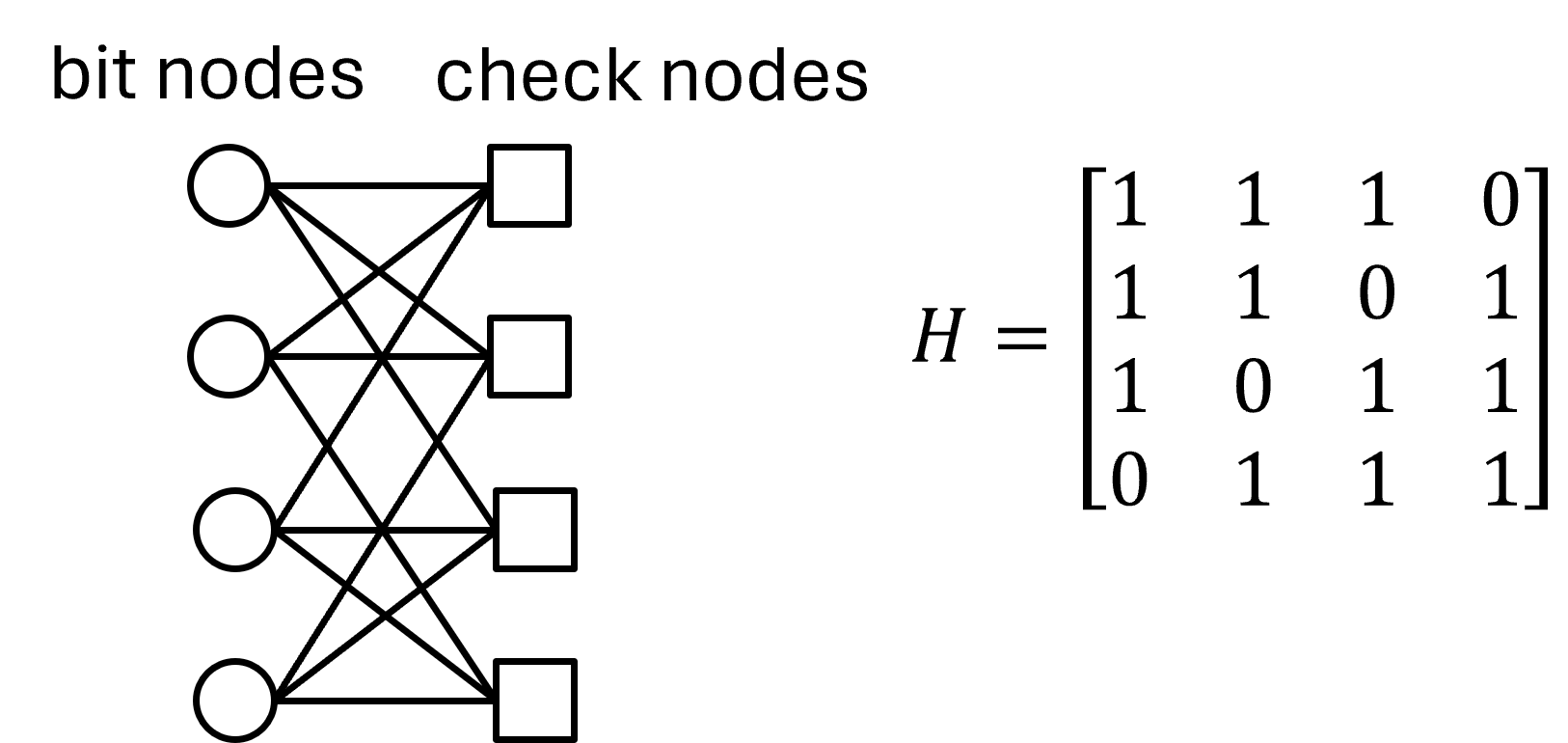}
        \caption{One example of Tanner graph and the corresponding parity-check matrix.}
\end{figure}

\begin{Def}[Graph product \cite{TZ13}]\label{def:graph-prod}
Let $\cG_1 $ and $\cG_2 $ be two graphs. 
The product of $\cG_1 $ and $\cG_2 $, denoted as $\cG_1\boxtimes\cG_2$, has a vertex set consisting of pairs $(x,y)$, where $x$ is a vertex from $\CGG_1$ and $y$ from $\CGG_2$.
Edges in the product graph connect two vertices $(x,y)$ and $(x',y')$ is either $x=x
'$ and $\{y,y'\}$ is an edge in $\CGG_2$ or $y=y'$ and $\{x,x'\}$ is an edge in $\CGG_1$.
\end{Def}
In particular, if $\cG_1= \left(V_1, W_1, E_1\right)$ and $\cG_2=\left(V_2, W_2, E_2\right)$ are two bipartite graphs,
then the product of $\cG_1$ and $\cG_2$ is also a bipartite graph $\cG_1\boxtimes\cG_2=(V,W,E) $,
where
\begin{align}\label{250111eq1}
\begin{aligned}
& V = V_1 \times V_2 \cup W_1 \times W_2,  \\
& W = W_1 \times V_2 \cup V_1 \times W_2, 
\end{aligned}
\end{align}
and
a vertex $(x_1,x_2)\in V$ connects to a vertex $(y_1,y_2)\in W$ iff:

1. $x_1=y_1$, and $x_2$ connects to $y_2$ in $\CGG_2$; or

2. $x_2=y_2$, and $x_1$ connects to $y_1$ in $\CGG_1$.

\begin{figure}[H]
    \includegraphics[width=0.9\linewidth]{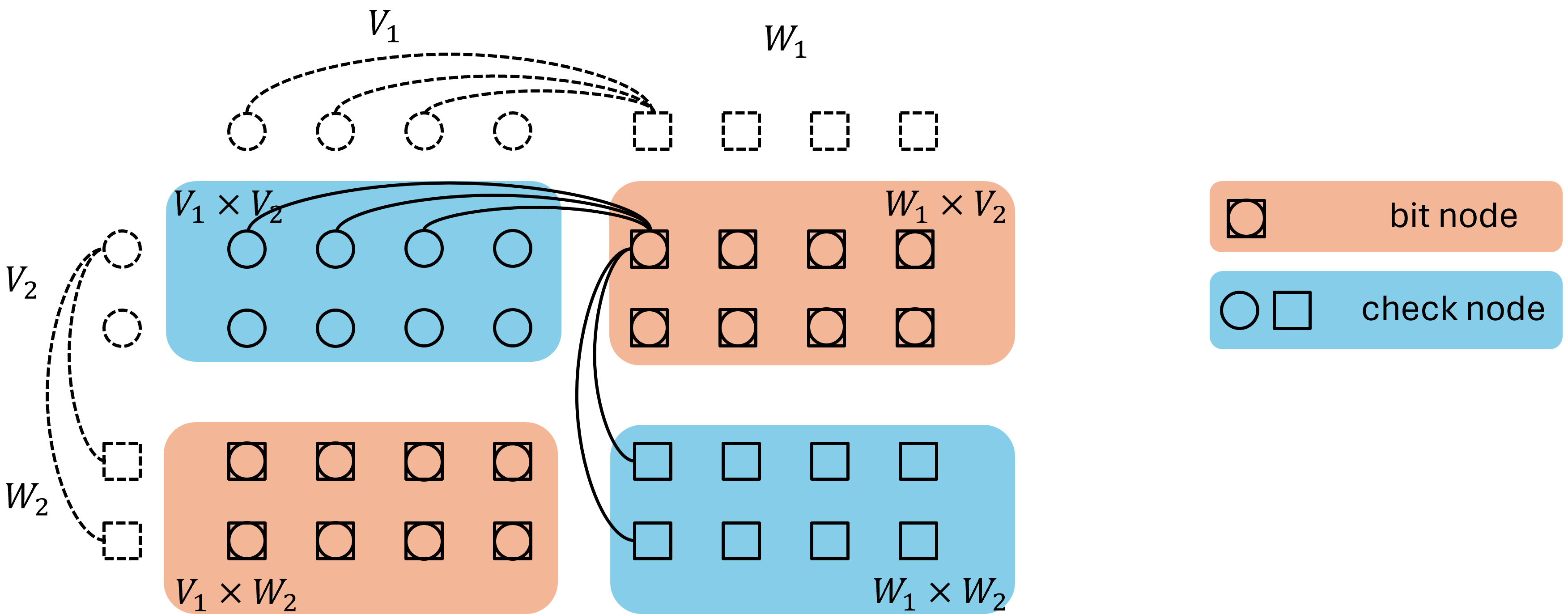}
      \caption{Bipartite graph product.}
\end{figure}

\begin{Def}\label{250101def1}
A bipartite graph $\cG= \left(V , W , E \right)$ is called  $(r,t,s)$-exact  if the following conditions are satisfied:

1. $\CGG$ is $(t,r+1)$-regular, i.e., every vertex in $V$ connects to $t$ vertices in $W$, and every vertex in $W$ connects to $r+1$ vertices in $V$.

2. Every two vertices in $W$ either have no common neighbors or have exactly $s+1$ common neighbors.

3. Every three vertices in $W$ either have no common neighbors or have exactly $1$ common neighbors.
\end{Def}

Figure \ref{241221fig1} gives an  $(1,2,1)$-exact Tanner graph and its parity-check matrix. 
\begin{figure}[H]
\centering
\includegraphics[width=0.6\linewidth]{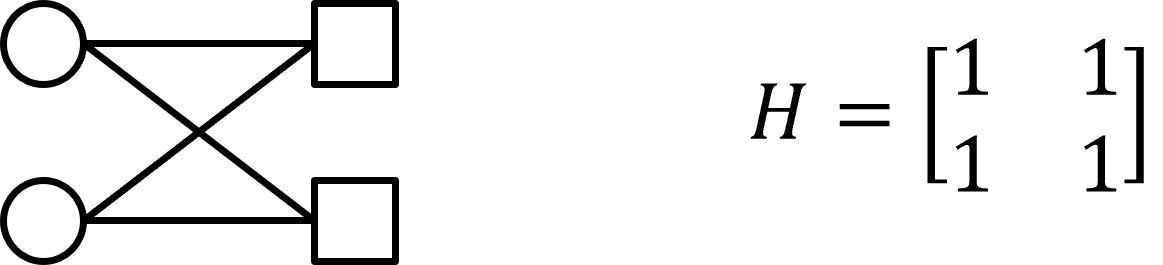}
\caption{An  $(1,2,1)$-exact  Tanner graph and the corresponding parity-check matrix.}\label{241221fig1}
\end{figure}

In fact, the exact Tanner graph corresponds to the exact classical LRC. 
\begin{prop}\label{250111prop1}
Given a classical code $\CCC$ with parity-check matrix $H$. 
If the Tanner graph $\CGG$ of $H$ is $(r,t,s)$-exact, then any subcode of $\CCC$ is an exact $(r,t,s)$-LRC.
\end{prop}
\begin{proof}
Consider $\CGG=([n],W,E)$.
Let us first figure out the family $ \left(\Gamma_l(j)\right)_{1\le l\le t,\; 1\le j\le n}$ of subsets of $[n]$ that satisfies the exact $(r,t,s)$-condition.
As $\CGG$ is $(r,t,s)$-exact,
every bit node $j$ in $[n]$ connects to $t$ check nodes in $W$. 
Therefore there are $t$ rows of $H$, say $ \vec h_1(j),...,\vec h_t(j)$,
whose $j$-th position is nonzero.
Let $\Gamma_l(j)$ be the support of $\vec h_l(j)$, $\forall l\in \set{1,...,t}$.
In particular, $j\in \Gamma_l(j)$.
Since $\CGG$ is $(r,t,s)$-exact,
each $ |\Gamma_l(j)|=r+1$. 
Moreover, $|\Gamma_{l_1}(j)\cap \Gamma_{l_2}(j)| = s+1$ when $l_1$ and $ l_2$ are different,
and
$\Gamma_{l_1}(j)\cap \Gamma_{l_2}(j) \cap \Gamma_{l_3}(j) = \{j\}$ when $l_1 , l_2$ and $l_3$ are different.
Hence $ \left(\Gamma_l(j)\right)_{1\le l\le t,\; 1\le j\le n}$  satisfies the exact $(r,t,s)$-condition.
 
Now for $1\le l\le t$, $\vec h_l(j)$ is a parity check of the code $\CCC$, hence any code word $\vec c\in \CCC$ satisfies $ \vec c\cdot \vec h_l(j)=0$,
which means
\[\sum_{i\in \Gamma_l(j)} \lambda_i \vec c(i)=0,\]
for some nonzero numbers $\lambda_i$.
In particular, 
\[\vec c(j) = -\frac{1}{\lambda_j} \sum_{i\in \Gamma_l(j) \setminus \{j\}} \lambda_i \vec c(i), \]
which means $\vec c(j)$ is determined by $\vec c|_{\Gamma_l(j) \setminus \{j\}}$.
Hence if $\vec c_1$, $\vec c_2$ are two code words in any subcode $\CCC_1$ of $ \CCC$ and $\vec c_1(j) \neq \vec c_2(j)$, then
\[\vec c_1|_{\Gamma_l(j) \setminus \{j\} } \neq \vec c_2|_{\Gamma_l(j) \setminus \{j\} }, \quad l\in\set{1,2,...,t}.\]
Hence $\CCC_1$ is exact $(r,t,s)$-LRC.
\end{proof}

\begin{prop}\label{250111prop2}
Suppose $H$ is a binary matrix whose Tanner graph $\CGG$ is $(r,t,s)$-exact, where $s+1$ and $r+1$ are both even.
Then the CSS quantum code $\CCC= CSS(H_X,H_Z)$ with $H_X=H_Z= H$ is an $(r,t,s)$-qLRC.
\end{prop}
\begin{proof}
Since $s+1$ and $r+1$ are both even and $\CGG$ is $(r,t,s)$-exact,
every two rows in $H$ have inner product 0, i.e.,
$H_X H_Z^T=0$,
hence $\CCC$ is a CSS quantum code.
In the following, we show $\CCC$  is an $(r,t,s)$-qLRC.

Say $\CGG=([n],W,E)$.
Define $ \left(\Gamma_l(j)\right)_{1\le l\le t,\; 1\le j\le n}$ as same as in the proof of Proposition \ref{250111prop1}.
Then by the same reason, $ \left(\Gamma_l(j)\right)_{1\le l\le t,\; 1\le j\le n}$ is a family of subsets of $[n]$ satisfying the exact $(r,t,s)$-condition.  
Moreover,
for every qubit $j\in [n]$ there are $t$ rows of $H$, say $ \vec h_1(j),...,\vec h_t(j)$,
whose supports are  $\Gamma_1(j)$,..., $\Gamma_t(j)$ respectively.
To show $\CCC$ is an exact $(r,t,s)$-qLRC,
we need to construct a recovery channel $\mathrm{Rec}_{l,j}$ for every $l\in [t]$ and $j\in [n]$,
such that each $\mathrm{Rec}_{l,j}$ is a quantum channel supported on qubits in $\Gamma_l(j)  $, and for every code state $\psi \in \CCC $ and every qubit $j\in [n]$ we have
$$\mathrm{Rec}_{l,j}  \left(\trace _j(\psi )\right) =\psi,\quad \forall 1\le l\le t. 
$$

As $\CCC$ is a CSS code with $H_X=H_Z= H$,
the stabilizer group of $\CCC$ contains stabilizers $ X^{\vec h_l(j)}$ and $  Z^{\vec h_l(j)}$ for any $j\in [n]$ and $l\in [t]$.
Then we construct the recovery channel $\mathrm{Rec}_{l,j}$ as follows (which is essentially Proposition 33 in \cite{GG23}). 
We first perform syndrome measurements for $X^{\vec h_l(j)} $ and $Z^{\vec h_l(j)}$.
This measurement will collapse the errored code word $\mathrm{Tr}_j (\psi)$ to a common eigenspace of $X^{\vec h_l(j)}$ and $Z^{\vec h_l(j)}$, and the measurement outcomes give the eigenvalues of the projected state for $X^{\vec h_l(j)}$ and $Z^{\vec h_l(j)}$.
Now,
\[\mathrm{Tr}_j (\psi) = \frac14 ( \psi + X_j \psi X_j + Y_j \psi Y_j + Z_j\psi Z_j).\]
Note that  $\psi$, $ X_j \psi X_j$, $ Y_j \psi Y_j $ and $Z_j\psi Z_j$ are all pure common eigenstates of $X^{\vec h_l(j)}$ and $Z^{\vec h_l(j)}$,
with eigenvalues $(1,1)$, $(1,-1)$, $(-1,-1)$ and $(-1,1)$ respectively. 
Therefore, the syndrome measurement outcomes provide the error $E=X_j^a Z_j^b$, and the recovery channel $\mathrm{Rec}_{l,j}$ then applies $E^\dag=Z_j^{-b} X_j^{-a}$ to revert the error and recover the original code state $\psi$.
Hence $\mathrm{Rec}_{l,j}$ is a quantum channel supported on  $\Gamma_l(j)  $  and can recover any code state $\psi \in \CCC $ from $\mathrm{Tr}_j (\psi)$.
Therefore by definition $\CCC$ is an $(r,t,s)$-qLRC.
\end{proof}

\begin{Examp}\label{241231exm1}
The following is a check matrix of the Hamming $(7,4,3)$-code.
\begin{align*}
H=\left[
\begin{matrix}
1 & 1 & 1 & 1 & 0 & 0 & 0\\
1 & 1 & 0 & 0 & 0 & 1 & 1\\
1 & 0 & 1 & 0 & 1 & 0 & 1\\
1 & 0 & 0 & 1 & 1 & 1 & 0\\
0 & 1 & 1 & 0 & 1 & 1 & 0\\
0 & 1 & 0 & 1 & 1 & 0 & 1\\
0 & 0 & 1 & 1 & 0 & 1 & 1\\
\end{matrix}
\right]
\end{align*}
In this matrix, every row has 4 nonzero elements, 
every column has 4 nonzero elements, 
every two rows have exactly 2 common nonzero positions,
and every three rows have either 1 or 0 common nonzero positions.
This configuration establishes a classic exact $(3,4,1)$-LRC.
The corresponding bipartite graph $\CGG$ is also $(3,4,1)$-exact.
Moreover,
one can check that $H$ satisfies $H H^T =\mathbf{0}$, hence we can choose $H_X=H_Z=H$ to construct a CSS code $Q(\CGG)$,
which is known as the quantum Hamming (7,1,3)-code.
From Proposition \ref{250111prop2},
$Q(\CGG)$ is an exact $(3,4,1)$-qLRC.
\end{Examp}

Now, let us proceed to construct larger exact LRC/qLRCs from known codes, supported by the following theorem. Before we delve into that, let us introduce a notation:
for a bipartite graph $\CGG=(V, W, E)$,  we denote $\CGG^T$ to be its transpose, i.e., $\CGG^T $ is the bipartite graph $ (W, V, E)$.

\begin{thm}\label{250101thm2}
Suppose $\CGG_1$ and $\CGG_1^T$ are both  $(r_1,t_1=r_1+1,1)$-exact bipartite graphs,
and $\CGG_2$ and $\CGG_2^T$ are both $(r_2,t_2=r_2+1,1)$-exact bipartite graphs.
Then $\CGG_1 \boxtimes \CGG_2$ and its transpose are both $(r=r_1+r_2+1,t=r_1+r_2+2,1)$-exact bipartite graphs.
\end{thm}
\begin{proof}
Say $\CGG_1= (V_1, W_1,E_1) $, $\CGG_2= (V_2, W_2,E_2) $, and denote $\CGG:= \CGG_1 \boxtimes \CGG_2= (V, W,E) $.
To show $\CGG$ and $\CGG^T$ are both $(r,t,1)$-exact, we need to check the conditions in Definition \ref{250101def1}.

First we show $\CGG$ is $(t,r+1)$-regular.
Pick any vertex $v\in V$, then either $v\in V_1\times V_2$ or $v\in W_1\times W_2$.
If $v=(v_1,v_2)\in V_1\times V_2$, then the neighboors of $v$ are the vertices $(v_1,c_2) $ and $(c_1,v_2)$, where $c_1$ connects $v_1$ in $\CGG_1$ and $c_2$ connects $v_2$ in $\CGG_2$. 
Since $\CGG_1$ is $(r_1,t_1,1)$-exact and $\CGG_2$ is  $(r_2,t_2,1)$-exact,
the vertex $v$ connects $t_1+t_2 = t$ vertices in $W$.
If $v=(c_1,c_2)\in W_1\times W_2$,
similarly the vertex $v$ connects $r_1+1+r_2+1 = t$ vertices in $W$.
On the other hand,
pick any vertex $c\in W$, then either $c\in V_1\times W_2$ or $c\in W_1\times V_2$.
If $c=(v_1,c_2)\in V_1\times W_2$, then the neighboors of $c$ are the vertices $(v_1,v_2) $ and $(c_1,c_2)$, where $c_1$ connects $v_1$ in $\CGG_1$ and $c_2$ connects $v_2$ in $\CGG_2$.
Since $\CGG_1$ is $(r_1,t_1,1)$-exact and $\CGG_2$ is  $(r_2,t_2,1)$-exact,
the vertex $c$ connects $t_1+r_2+1 = r+1$ vertices in $V$.
If $c=(c_1,v_2)\in W_1\times V_2$,
similarly the vertex $v$ connects $r_1+1+t_2 = r+1$ vertices in $V$.
Hence $\CGG$ is $(t,r+1)$-regular.

Next, we show that for every pair of vertices $c$ and $c'$ in $W$,
if they have common neighbors,
then they must have exactly $2$ common neighbors.
We have 4 cases.

Case 1, if $c=(c_1,v_2)$ and $c'=(c_1',v_2')$ are both from $W_1\times V_2$. 
If  $c_1= c_1'$ and $v_2\neq v_2'$,
from the condition that $\CGG_2^T$ is $ (r_2, t_2, 1)$-exact we have that, either $v_2$ and $v_2'$ have no common  neighbors in $\CGG_2$, which implies that $c$ and $c'$ have no common neighbors;
or $v_2$ and $v_2'$ have exactly 2 common  neighbors, say $c_+$ and $c_-$,
which implies that $c$ and $c'$ also have exactly 2 common  neighbors, $ (c_1,c_+)$ and $(c_1,c_-)$.
If  $c_1\neq c_1'$ and $v_2= v_2'$,
from the condition that $\CGG_1$ is $ (r_1, t_1, 1)$-exact,
we also have that $c$ and $c'$ must have exactly 2 common  neighbors.
Finally, if $c_1\neq c_1'$ and $v_2\neq v_2'$,
then based on the definition of the graph product, it is direct to see that $c$ and $c'$ have no common neighbors.

Case 2, if $c=(v_1,c_2)$ and $c'=(v_1',c_2')$ are both from $V_1\times W_2$.
This case is similar to Case 1. From the conditions that $\CGG_1^T$ is $ (r_1, t_1, 1)$-exact and $\CGG_2$ is $ (r_2, t_2, 1)$-exact,
we derive that $c$ and $c'$  either have no common neighbors or have exactly $2$ common neighbors.

Case 3, if $c=(c_1,v_2)\in W_1\times V_2$ and $c'=(v_1',c_2')\in V_1\times W_2$. 
Then from the definition of graph product,
if $c$ and $c'$ has at least one common neighbor,
then they must have exactly 2 common neighbors, which are $(c_1,c_2')$ and $(v_1',v_2)$.

Case 4, if $c=(v_1,c_2)\in V_1\times W_2$ and $c'=(c_1',v_2')\in W_1\times V_2$.
This case is similar to Case 3.

Hence every pair of vertices in $W$ either have no common neighbors or have exactly $2$ common neighbors in $\CGG$.

Finally, we show that every triple of different vertices $c$, $c'$, and $c''$ in $W$ have at most $1$ common neighbor.
If $c$ and $c'$ do not have common neighbors, then we are done.
So in the following, we assume $c$ and $c'$ have exactly $2$ common neighbors in $\CGG$.
Again we divide it into 4 cases.

Case 1', if $c=(c_1,v_2)$ and $c'=(c_1',v_2')$ are both from $W_1\times V_2$.
If  $c_1= c_1'$ and $v_2\neq v_2'$,
by assumption $v_2$ and $v_2'$ must have exactly 2 common  neighbors, say $c_+$ and $c_-$, and the 2 common  neighbors of $c$ and $c'$ are $ (c_1,c_+)$ and $(c_1,c_-)$.
Now if $c'' = (c_1'',v_2'')$ connects to both $ (c_1,c_+)$ and $(c_1,c_-)$,
then $c_1''=c_1$, and $v_2''$ also connects to both $c_+$ and $c_-$ in $\CGG_2$.
Hence $v_2$ $v_2'$ and $v_2''$ have 2 common  neighbors, contradicts to the $ (r_2, t_2, 1)$-exact condition for $\CGG_2^T$.
Hence $c$, $c'$ and $c''$ have at most $1$ common neighbor.
If  $c_1\neq c_1'$ and $v_2= v_2'$,
from the condition that $\CGG_1$ is $ (r_1, t_1, 1)$-exact,
we also have that $c$, $c'$ and $c''$ have at most $1$ common neighbor.

Case 2', if $c=(v_1,c_2)$ and $c'=(v_1',c_2')$ are both from $V_1\times W_2$.
This case is similar to Case 1'.  

Case 3', if $c=(c_1,v_2)\in W_1\times V_2$ and $c'=(v_1',c_2')\in V_1\times W_2$. 
From the discussion in case 3,
the 2 common neighbors of $c$ and $c'$ must be $(c_1,c_2')$ and $(v_1',v_2)$.
Now by the definition of graph product, there does not exist another vertex $c''\in W$ that can connect to both $(c_1,c_2')$ and $(v_1',v_2)$.
Hence $c$, $c'$ and $c''$ have at most $1$ common neighbor.

Case 4', if $c=(v_1,c_2)\in V_1\times W_2$ and $c'=(c_1',v_2')\in W_1\times V_2$.
This case is similar to Case 3'.

Therefore every triple of vertices in $W$ either have no common neighbors or have exactly 1 common neighbors,
so $\CGG$ is $(r,t,1)$-exact.
Finally,
$\CGG^T= \CGG_1^T \boxtimes \CGG_2^T$,
hence our assumptions also imply that $\CGG^T$ is $(r,t,1)$-exact.
\end{proof}

\begin{Rem}
Here let us compare our construction of quantum codes in Proposition \ref{250111prop2} 
with the hypergraph product code introduced in \cite{TZ13}.
Recall the hypergraph product code (HPC) is a kind of quantum CSS code constructed from bipartite graph products.
Using the notations in equation \eqref{250111eq1},
the HPC corresponding to the product of $\CGG_1$ and $\CGG_2$ is a CSS code $CSS(H_X',H_Z')$,
where the Tanner graph of $H_X'$ is the subgraph of $\CGG_1\boxtimes \CGG_2$ with set of bit nodes $V$ and set of check nodes $W_1\times V_2$,
and the Tanner graph of $H_Z'$ is the subgraph of $\CGG_1\boxtimes \CGG_2$ with set of bit nodes $V$ and set of check nodes $V_1\times W_2$.
The graph product ensures that $H_X' H_Z'^T=0$,
hence $CSS(H_X',H_Z')$ is a CSS code.
In Proposition \ref{250111prop2},
we consider the CSS code $CSS(H_X,H_Z)$ with $H_X=H_Z=H$ whose Tanner graph is $(r,t,s)$-exact.
By the $(r,t,s)$-exact condition of the Tanner graph (with even $s+1$ and even $r+1$), we have that $CSS(H_X,H_Z)$ is a CSS code.
By Theorem \ref{250101thm2},
the needed Tanner graph can be constructed from a graph product of two bipartite graphs with certain $(r,t,s)$-exact condition.
In particular, when the Tanner graph of $H$ is $\CGG_1\boxtimes \CGG_2$,
we have
\begin{align*}
H_X=H_Z=H=\left(
\begin{aligned}
H_X'\\
H_Z'
\end{aligned}\right).
\end{align*}

\end{Rem}

\begin{thm}\label{250112thm1}
For $i=1,2$,
suppose $r_i$ is odd, $t_i$ is even, and $\cG_i=(V_i, W_i, E_i)$ is an  $(r_i,t_i,1)$-exact Tanner graph whose transpose $\CGG_i^T$ is $(t_i-1,r_i+1,1)$-exact.
Let $\cC_1=(n_1, k_1, d_1)$, $\cC_2=(n_2,k_2,d_2)$,  $\cC_1^T=(n_1^T, k_1^T, d_1^T)$, and $\cC_2^T=(n_2^T,k_2^T,d_2^T)$ be binary classic codes whose parity check matrices correspond to the Tanner graphs $\cG_1$, $\cG_2$, $\cG_1^T$ and $\cG_2^T$ respectively.
Let   $\cG=\cG_1\boxtimes\cG_2$ 
and assume $Q(\CGG)$ is an  $[[n, k, d]]$-code.
Then we have the following result:  

(1). $n=|W_1||W_2|+|V_1||V_2|$.

(2). $k\ge (2k_1-n_1)(2k_1^T-n_1^T)+(2k_2-n_2)(2k_2^T-n_2^T)$.

(3). $d \ge \min \left\{d_1, d_2, d_1^T, d_2^T\right\}$.

(4). If $t_i=r_i+1$,
then $Q(\CGG)$ is an exact $(r=r_1+r_2+1,t=r_1+r_2+2,1)$-qLRC.
\end{thm}
\begin{proof}
Part (4) comes directly from Theorem \ref{250101thm2}.
The proof of (1)-(3) involves several technical mathematical lemmas; therefore, we place the complete proof in Appendix~\ref{appendix1} for detailed reference.
\end{proof}

\begin{Rem}
Here we remark that $2k_1-n_1$, $2k_1^T-n_1^T$, $2k_2-n_2$, $2k_2^T-n_2^T$ are all non-negative.
For instance for $2k_1-n_1$, 
from the exact conditions of $\cG_1 $, the parity check matrices $H_1$ of $\CCC_1$ is self-orthogonal, i.e., $H_1H_1^T=0$. 
Hence $\mathrm{rank}(H_1) \le n_1/2$,
and $k_1= n_1- \mathrm{rank}(H_1) \ge n_1/2$.
\end{Rem}
\begin{Rem}
Here we remark that Theorem \ref{250112thm1} is different from  Proposition 14 and Theorem 15  in \cite{TZ13},
which compute the dimension $k_Q$ and minimal distance $d_Q$ of a HPC.
The quantum code $Q(\CGG)$ considered in Theorem \ref{250112thm1} is not a HPC.
The dimension of $Q(\CGG)$  is $k\ge (2k_1-n_1)(2k_1^T-n_1^T)+(2k_2-n_2)(2k_2^T-n_2^T)$,
where $2k_1-n_1$, $2k_1^T-n_1^T$, $ 2k_2-n_2$ and $2k_2^T-n_2^T$ are respectively the numbers of logical qubits in four quantum codes that related to the given Tanner graphs,
while in \cite{TZ13},
the dimension of the HPC  is $k\ge k_1'k_1'^T+ k_2'k_2'^T$,
where $k_1'$, $k_1'^T$, $ k_2'$, and $k_2'^T$ are dimensions of four classic codes that relate to the given Tanner graphs.
\end{Rem}

In the following example, we construct binary exact $(r,t,1)$-LRC/qLRC from the bipartite graph product of two exact bipartite graphs,
where $n$ can be arbitrarily large.

\begin{Examp}
Consider the Tanner graph of the parity-check matrix $H$, denoted as $\CGG_1$ in Example \ref{241231exm1}.
Recall that the corresponding classic code has parameters $n=7$, $k=4$, $d=3$, $r=3$, $t=4$ and $s=1$,
and $\mathrm{rank}(H)=3$.
We inductively define $\CGG_m := \CGG_{m-1} \boxtimes \CGG_1$ for $m\ge 2$.
According to Theorem \ref{250101thm2},  
$\CGG_m$ forms $(r_m, r_m+1,1)$-exact bipartite graphs,
where $r_m = 2^{m+1}-1$.

Let $Q(\CGG_m)$ denote the quantum CSS code $ \mathrm{CSS}(\CCC_X,\CCC_Z)$  for $\CCC_X=\CCC_Z = \ker H_m $, where $H_m$ is the parity-check matrix whose Tanner graph is $\CGG_m$.
Then $\CGG_m$ is an exact $(2^{m+1}-1, 2^{m+1},1)$-qLRC.
Moreover,
the parameters of $\CGG_m$ are
$n = 7^{2^{m-1}}2^{2^{m-1}-1},\quad k \ge  2^{2^{m-1}-1}, \quad d\ge 3.$

\end{Examp}

\section{Implication on classical LRCs}

In this section, we will apply the method developed for quantum locally recoverable codes (qLRCs) to derive improved bounds for classical locally recoverable codes (LRCs). Specifically, we will provide an upper bound on the code dimension 
$k$ as a function of the parameters 
$n,r,t$, and the code distance $d$. To accomplish this, we introduce the following lemma.

First, let us denote
\[p_1(r,t) =  1-\frac{1}{\prod_{l=1}^t (1+\frac1{lr})} .\]

\begin{lem}\label{241028lem3}
Let $n , t ,r,d$ be positive integers with $n\ge \max\{ r+1,d-1\}$.
Suppose  every $j\in [n]$ is assigned with $t$ subsets $\Gamma_1(j),..., \Gamma_t(j) \subseteq [n]$, such that $j\in \Gamma_\alpha(j)$, $|\Gamma_\alpha(j)|\le r+1$ and $\Gamma_{\alpha }(j)\cap \Gamma_{\beta}(j) = \{j\} $  for different $\alpha, \beta$.
	Then we can find a sequence of numbers $j_1,...,j_M\in [n]$ with 
	\[M =(d-1) + N_1(n,r,d, \left\lceil n p_1(r,t) \right\rceil) \]
	such that,
	for every $d\le k\le M$, there is at least one $\alpha$ satisfying
	\begin{align}\label{241023eq1}
		\Gamma_\alpha(j_k) \cap \{j_{1},j_{2},...,j_{k-1}\} = \emptyset.
	\end{align}
\end{lem}
\begin{proof}
By Lemma 1 in \cite{TB16},
we can find a sequence of numbers
$i_1,...,i_{M'}\in [n]$ with 
$M' = \left\lceil n p_1(r,t) \right\rceil$
	such that,
	for every $1\le k\le M'$, there is an $\alpha_k$ satisfying
	\begin{align}\label{241023eq1}
		\Gamma_{\alpha_k}(i_k) \cap \{i_{1},i_{2},...,i_{k-1}\} = \emptyset.
	\end{align}
Now we apply Lemma \ref{241227lem2} on the sets $\Gamma_{\alpha_k}(i_k)$, $k=1,...,M'$, with $N := N_1(n,r,d,M')$.
Since $\Gamma_{\alpha_k}(i_k)$ may have size smaller than $r+1$, we denote $A_k$ to be a subset of $[n]$ such that $\Gamma_{\alpha_k}(i_k)\subseteq A_k $ and $|A_k|=r+1$ for each $k$. 
Then by Lemma \ref{241227lem2},
we can find a subset $S\subseteq [M']$ with size $|S|=N$ such that
\[\left| \bigcup_{m\in S} A_m\right| \le N(r+1) - \frac{N(N-1)}{M'(M'-1)}(M'(r+1)-n) \le n-(d-1),\]
where the last inequality is from the definition of $N_1(n,r,d,M')$.
Now let $j_1,...,j_{d-1}$ be $(d-1)$ elements in $[n]\setminus \bigcup_{m\in S} A_m$,
then for every $k\in S$, 
$$\Gamma_{\alpha_k}(i_k) \cap \{j_{1},j_{2},...,j_{d-1}\} = \emptyset.$$
Moreover, let $ j_d,j_{d+1},..., j_M$ be all the elements in $S$ (ordered as in the sequence $i_1,...,i_{M'}$.)
Then for every $d\le k\le M$, there is at least one $\alpha$ satisfying
	\begin{align*} 
		\Gamma_\alpha(j_k) \cap \{j_{1},j_{2},...,j_{k-1}\} = \emptyset.
	\end{align*}
The proof is completed.
\end{proof}

In the classic case, we have the following result.
\begin{thm}\label{250101thm1}
Suppose $n,k,d,r,t\ge 1$ are given numbers with $n\ge \max\{r+1,d-1\}$.
If
$\CCC$ is a classic $(r,t)$-locally recoverable $(n,k,d)$-code,
then
\[k\le   n-(d-1) - N_1(n,r,d, \left\lceil n p_1(r,t) \right\rceil). \]
\end{thm}
\begin{proof}
Since $\CCC$ is a  classic $(r,t)$-locally recoverable code and $n\ge \max\{r+1,d-1\}$,
from Lemma \ref{241028lem3},
we can find a sequences of numbers, $j_1,...,j_{M}\in [n]$, where  
\[M =(d-1) + N_1(n,r,d, \left\lceil n p_1(r,t) \right\rceil), \]
such that,
for every $d\le m\le M$ there is an $\alpha_m$ satisfying
\begin{align*} 
\Gamma_{\alpha_m}(j_m) \cap \{j_{1},j_{2},...,j_{m-1}\} = \emptyset.
\end{align*}
Hence any error on the bits $j_1,...,j_{M}$ can be recovered,
that is,
when two code words $\vec c_1$ and $\vec c_2$ are the same on $[n]\setminus \{ j_1,...,j_{M}\} $,
then $\vec c_1= \vec c_2$.
Denote $A = [n]\setminus \{ j_1,...,j_{M}\}$.
From the above analysis,
for every $\vec v\in \Z_q^A $  there is at most one code word $\vec c\in \CCC$ such that $\vec c|_A = \vec v$.
Hence $|\CCC| \le |\Z_q^A |$,
which implies that
\[k = \log_q|\CCC| \le |A| = n-(d-1) - N_1(n,r,d, \left\lceil n p_1(r,t) \right\rceil).\]
\end{proof}

Let us compare our bound of classic LRCs from Theorem \ref{250101thm1} with some known bounds.
Let $\mathcal{C}$ be a classic $(n, k, d, r)$-locally recoverable code.
A generalized Singleton bound was proposed in \cite{GHSY12}:
\begin{equation}\label{eqn:250101-1}
d \leq n-k-\left\lceil\frac{k}{r}\right\rceil+2.
\end{equation}
Furthermore, an improved bound for classic $(n, k, d, r)$-LRC with availability $t$ is given in \cite{BTV17}:
\begin{equation}\label{eqn:250101-2}
d \leq n-\sum_{i=0}^t\left\lfloor\frac{k-1}{r^i}\right\rfloor .
\end{equation}

In Figure \ref{fig:comp-thm10}, we compare the upper bounds of $k$ among Theorem \ref{250101thm1}, generalized Singleton bound (i.e., Equation~\eqref{eqn:250101-1}), and the bound in \cite{BTV17} (i.e., Equation~\eqref{eqn:250101-2}).
Based on the figures, it is direct to see that Theorem \ref{250101thm1}
provides a tighter upper bound than the previous results. 

\begin{figure}[h]
    \centering
        \begin{minipage}{0.48\textwidth}
        \centering
        \includegraphics[width=\textwidth]{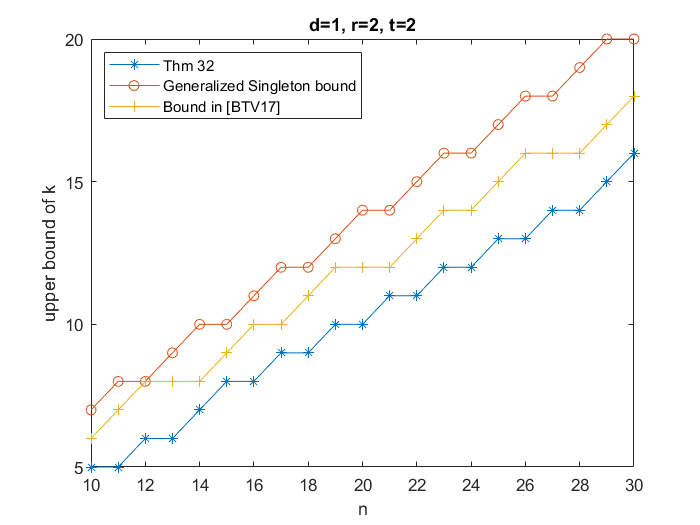}
    \end{minipage}\hfill
    \begin{minipage}{0.48\textwidth}
        \centering
        \includegraphics[width=\textwidth]{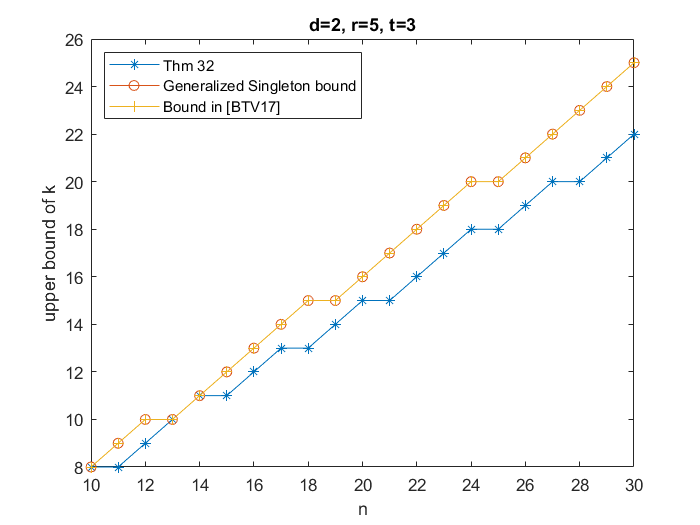}
    \end{minipage}
    
    \vskip\baselineskip

    \begin{minipage}{0.48\textwidth}
        \centering
        \includegraphics[width=\textwidth]{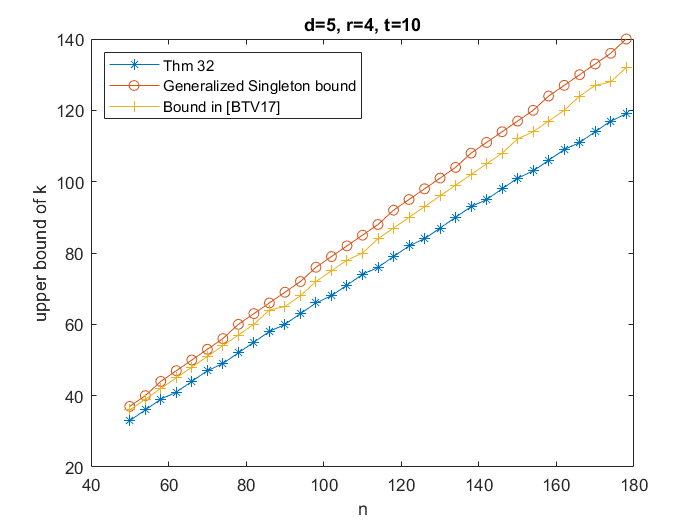}
    \end{minipage}\hfill
    \begin{minipage}{0.48\textwidth}
        \centering
        \includegraphics[width=\textwidth]{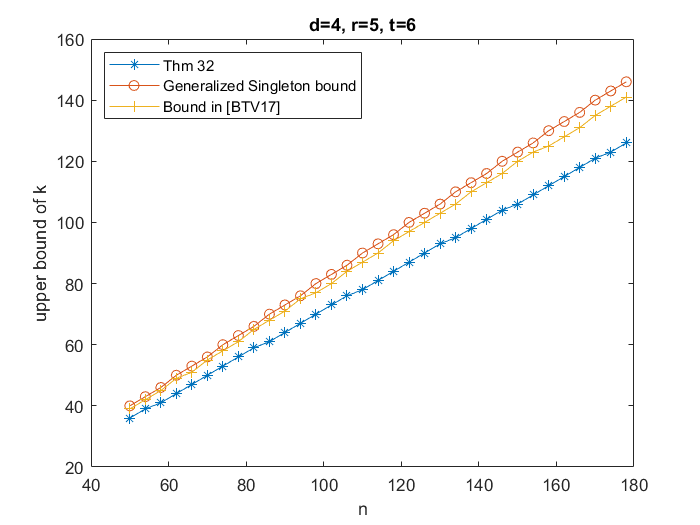}
    \end{minipage}
    \caption{Comparison of the upper bounds of $k$ from Theorem \ref{250101thm1}, Equation \eqref{eqn:250101-1}, and Equation \eqref{eqn:250101-2}.The orange line represents   upper bounds of $k$ computed using Equation \eqref{eqn:250101-1} as $n$ varies;
the yellow line represents the upper bounds of $k$ computed using Equation \eqref{eqn:250101-2} as $n$ varies;
and the blue line represents  the upper bounds of $k$ given by Theorem \ref{250101thm1}.}
    \label{fig:comp-thm10}
\end{figure}

\section{Conclusion}
In this work, we have investigated qLRCs with intersecting recovery sets, a class of codes that offers enhanced flexibility in local recovery procedures. We derived new Singleton-like upper bounds on the parameters of these codes, establishing fundamental trade-offs between code length, dimension, and local recovery capabilities. Furthermore, we presented a construction method for qLRCs with intersecting recovery sets by introducing a variation of the hypergraph product, providing a concrete approach for realizing these codes.
In addition, we use our method to improve the results for the classical LRCs.

Several compelling avenues for future research emerge from this work.
A primary direction involves exploring the optimality of the derived bounds. Identifying quantum codes that saturate these bounds would not only demonstrate their tightness but also provide valuable insights into the limits of local recoverability in quantum information processing. This will be helpful in 
the development of 
new code design techniques and analytical tools.

Additionally, investigating the connection between qLRCs and physical systems is a promising direction. Specifically, exploring the ground states of local Hamiltonians that exhibit local recoverability properties could lead to new physical implementations of these codes.

Finally, 
as the quantum codes are usually related to the classification of quantum phases of matters, 
it would be valuable to explore the role of local recoverability within the context of quantum phases. 
This could potentially uncover new topological phases or provide novel characterizations of existing ones, enhancing our understanding of the role of locality in both QECC and quantum phases.

\section{Appendix}\label{appendix1}

Let us focus on calculating the lower bound of the number of logical qubits $k$ of the CSS code constructed by $\cG_1\boxtimes\cG_2$, where $\cG_1$ and $\cG_2$ are two given exact Tanner graphs.

\begin{lem}\label{lem:241221-1}
Given an $(r, t, 1)-$exact bipartite graph $\cG$ with odd $r$, let 
$H$ be the check matrix whose Tanner graph is $\cG$. Then, we have
$$\row(H)\subseteq \ker(H).$$
\end{lem}
\begin{proof}
Let $H_i$ be the $i$-th row in $H$. For any $i,j$, $i\neq j$, $\ep{H_i, H_j}=s+1=0 \mod 2$. For any $i$, $\ep{H_i, H_i} = r = 0 \mod 2$. Thus, $\row(H)\subseteq \ker(H)$.
\end{proof}

In the following part of this section, we will focus on the product graph constructed from exact bipartite graphs, where $s=1$ and $r$ is odd.

\begin{prop}\label{prop:lb-rank-H}
Suppose $\cG_1=(V_1, W_1, E_1)$ and $\cG_2=(V_2,W_2,E_2)$ are $(r_i, t_i, 1)-$exact bipartite graphs, $\cG_1^T$ and $\cG_2^T$ are $(t_i-1, r_i+1, 1)-$exact bipartite graphs with odd $r_i$ and even $t_i$ for $i=1,2$.
Let $H_1$, $H_2$ and $H$ be the parity-check matrices whose Tanner graphs are $\cG_1$, $\CGG_2$ and $\CGG:=\cG_1\boxtimes\cG_2$ respectively.
Then we have 
\begin{align*}
    &|V|-2 \rank(H)\\
    \geq&\left(|V_1|-2\rank(H_1)\right)\left(|W_1|-2\rank(H_1^T)\right)+
\left(|V_2|-2\rank(H_2)\right)\left(|W_2|-2\rank(H_2^T)\right).
\end{align*}

\end{prop}

To prove this proposition, we first introduce some notations and state a few lemmas. 
Assume $\cG_i=(V_i, W_i, E_i)$ and $\CGG=(V,W,E)$,
where 
\begin{equation*}
V_i = \{v_1^i, v_2^i, \dots, v_{m_i}^i\},\quad
W_i = \{w_1^i, w_2^i, \dots, w_{n_i}^i\},
\end{equation*} 
with $m_i = |V_i|$ and $n_i = |W_i|$, for $i=1,2$. 
Let $\cC_1$ and $\cC_2$ be the binary codes specified by $\cG_1$ and $\cG_2$, respectively, 
with $H_1$ and $H_2$ being the corresponding check matrices.
The rows of $H_i$ are indexed by elements of $W_i$, and the columns of $H_i$ are indexed by elements of $V_i$, for $i=1,2$.

Given a set $S=\{s_1, \cdots,s_{|S|}\}$, let $\{e_1^S, e_2^S, \cdots, e_{|S|}^S\}$ be a set of standard orthonormal basis vectors for $\BFF_2^S$, 
where $e_j^S$ is an $|S|$-dimensional vector with $1$ in the $j$-th position and $0$ in all other positions. 
For convenience, let $V_{WW}=W_1\times W_2$, 
$V_{VV}=V_1\times V_2$, 
$W_{VW}=V_1\times W_2$, 
and $W_{WV}=W_1\times V_2$. 
Then, $V=V_{WW}\bigcup V_{VV}$ and $W=W_{VW}\bigcup W_{WV}$.

We can now describe the rows of the parity-check matrix $H$. 
The rows of $H$ are indexed by the elements of $W_{VW}\bigcup W_{WV}$, 
and the columns of $H$ are indexed by the elements of $V_{WW}\bigcup V_{VV}$. 
Note that 
$$\BFF_2^V=\BFF_2^{V_{WW}}\oplus\BFF_2^{V_{VV}}\cong\left(\BFF_2^{W_1}\otimes\BFF_2^{W_2}\right)\oplus\left(\BFF_2^{V_1}\otimes\BFF_2^{V_2}\right).$$
For example, if a check node $w\in W$ is only connected to the node $(w_i^1, w_j^2)\in V_{WW}$, then the corresponding row of the check node in $H$ is 
$$\vec{H}(w)=\left(e_i^{W_1}\otimes e_j^{W_2}, \;\;\mathbf{0}\right),$$
where 
$e_i^{W_1}\otimes e_j^{W_2}\in\BFF_2^{W_1}\otimes\BFF_2^{W_2}$, and $\mathbf{0}\in\BFF_2^{V_1}\otimes\BFF_2^{V_2}$.
Similarly, if a check node $w\in W$ is only connected to a node $(v_i^1, v_j^2)\in V_{VV}$, then the corresponding row of the check node in $H$ is
$$\vec{H}(w)=\left(\mathbf{0}, \;\;e_i^{V_1}\otimes e_j^{V_2}\right),$$
where 
$\mathbf{0}\in\BFF_2^{W_1}\otimes\BFF_2^{W_2}$, and $e_i^{V_1}\otimes e_j^{V_2}\in\BFF_2^{V_1}\otimes\BFF_2^{V_2}$.
Thus, for a given $(v_i^1, w_j^2)\in W_{VW}$, the corresponding row in $H$ is the vector
\begin{equation}
\left(\vec{H_1}(v_i^1)\otimes e_{j}^{W_2}, \;\;
e_{i}^{V_1}\otimes \vec{H_2}(w_j^2)\right),
\end{equation}
where $\vec{H_i}(w_j^i)$, $w_j^i\in W^i$ is the $j$-th row in $H_i$, and $\vec{H_i}(v_j^i)$, $v_j^i\in V^i$ is the $j$-th column in $H_i$.
Similarly, 
for a given $(w_i^1, v_j^2)\in W_{WV}$, the corresponding row in $H$ is the vector
\begin{equation}
\left(e_{i}^{W_1}\otimes \vec{H_2}(v_j^2), \;\;
\vec{H_1}(w_i^1)\otimes e_{j}^{V_2}\right).
\end{equation}

\begin{lem}\label{lem:241221-2}
With the notations in Proposition \ref{prop:lb-rank-H},
we have
$\row(H)\subseteq\ker(H)$, $\row(H_i)\subseteq\ker(H_i)$ and $\row(H_i^{T})\subseteq\ker(H_i^{T})$ for $i=1,2$.
\end{lem}
\begin{proof}
From Lemma \ref{lem:241221-1}, $\row(H_i)\subseteq\ker(H_i)$ and $\row(H_i^{T})\subseteq\ker(H_i^{T})$ for $i=1,2$.
To prove $\row(H)\subseteq\ker(H)$, we consider the following three cases.

Case 1. If $(v_i^1, w_j^2), (v_k^1, w_l^2)\in W_{VW}$. 
The inner product of the corresponding two rows in $H$ is
\begin{equation*}
\begin{array}{rl}
& \ep{\vec{H}(v_i^1, w_j^2), \vec{H}(v_k^1, w_l^2)} \\
=& \ep{\vec{H_1}(v_i^1), \vec{H_1}(v_k^1)}
\ep{e_{j}^{W_2}, e_{l}^{W_2}}
+ \ep{e_{i}^{V_1}, e_{k}^{V_1}}
\ep{\vec{H_2}(w_j^2), \vec{H_2}(w_l^2)}\\
=&0.
\end{array}
\end{equation*}

Case 2. If $(w_i^1, v_j^2), (w_k^1, v_l^2)\in W_{WV}$.
The inner product of the corresponding two rows in $H$ is
\begin{equation*}
\begin{array}{rl}
& \ep{\vec{H}(w_i^1, v_j^2), \vec{H}(w_k^1, v_l^2)} \\
=& \ep{e_{i}^{W_1}, e_{k}^{W_1}}
\ep{\vec{H_2}(v_j^2), \vec{H_2}(v_l^2)}
+ \ep{\vec{H_1}(w_i^1),\vec{H_1}(w_k^1)}
\ep{e_{j}^{V_2}, e_{l}^{V_2}} \\
=&0.
\end{array}
\end{equation*}

Case 3. If $(v_i^1, w_j^2)\in W_{VW}$ and $(w_k^1, v_l^2)\in W_{WV}$. 
Note that 
$\vec{H_1}(v_i^1)\cdot e_{k}^{W_1}$ and $e_i^{V_1}\cdot \vec{H_1}(w_k^1)$
are either both $0$ or both $1$.
Also, $e_j^{W_2}\cdot \vec{H_2}(v_l^2)$ and $\vec{H_2}(w_j^2)\cdot e_{l}^{V_2}$
are either both $0$ or both $1$.
The inner product of the corresponding two rows in $H$ is
\begin{equation*}
\begin{array}{rl}
& \ep{\vec{H}(v_i^1, w_j^2), \vec{H}(w_k^1, v_l^2)} \\
=& \ep{\vec{H_1}(v_i^1), e_{k}^{W_1}}
\ep{e_j^{W_2}, \vec{H_2}(v_l^2)}
+ \ep{e_i^{V_1}, \vec{H_1}(w_k^1)}
\ep{\vec{H_2}(w_j^2), e_{l}^{V_2}} \\
=&0.
\end{array}
\end{equation*}
Hence any two rows in $H$ have inner product 0.
So $\row(H)\subseteq\ker(H).$
\end{proof}

For matrices $H_1$ and $H_2$ defined in Proposition \ref{prop:lb-rank-H}, since $\row(H_j)\subseteq\ker(H_j)$ and $\row(H_j^{T})\subseteq\ker(H_j^{T})$ for $j=1,2$, one can choose $U_j$ and $U_j^{T}$ such that 
$$U_j\oplus\row(H_j)=\ker(H_j), \quad j=1,2,$$
$$U_j^{T}\oplus\row(H_j^{T})=\ker(H_j^{T}), \quad j=1,2.$$
We also naturally define two mappings:
\begin{equation}
\begin{aligned}
\phi:\; \BFF_2^{V_1}\otimes\BFF_2^{V_2}&\rightarrow \BFF_2^{V}=\left(\BFF_2^{W_1}\otimes\BFF_2^{W_2}\right)\oplus\left(\BFF_2^{V_1}\otimes\BFF_2^{V_2}\right)\\
\mathbf{a} &\mapsto (\mathbf{0}, \mathbf{a})
\end{aligned}   
\end{equation}
\begin{equation}
\begin{aligned}
\phi^{T}:\; \BFF_2^{W_1}\otimes\BFF_2^{W_2}&\rightarrow \BFF_2^{V}=\left(\BFF_2^{W_1}\otimes\BFF_2^{W_2}\right)\oplus\left(\BFF_2^{V_1}\otimes\BFF_2^{V_2}\right)\\
\mathbf{b} &\mapsto ( \mathbf{b},\mathbf{0})
\end{aligned}   
\end{equation}

With the notations above, we have the following lemma. 
\begin{lem}\label{lem:241221-3}
(1). $\phi(U_1\otimes U_2)\subseteq\ker(H)$\\
(2). 
$\phi^{T}(U_1^{T}\otimes U_2^{T})\subseteq\ker(H)$.
\end{lem}
\begin{proof}
(1). For any row $\vec{H}(*,*)$ in the matrix $H$, we need to show that $\ep{\vec{H}(*,*), (\mathbf{0}, u_1\otimes u_2)}=0$ for any $u_1\in U_1$ and $u_2\in U_2$.

For every $(v_i^1, w_j^2)\in W_{VW}$, the corresponding row in $H$ is $\vec{H}(v_i^1, w_j^2)$. 
Since $u_2\in U_2$, $u_2\in \ker(H_2)$, we know that $\ep{u_2, \vec{H_2}(w_j^2)}=0.$ 
Thus, 
\begin{equation*}
\begin{array}{rl}
& \ep{\vec{H}(v_i^1, w_j^2), (\mathbf{0}, u_1\otimes u_2)}\\
=& 0 + \ep{e_{i}^{V_1}, u_1}
\ep{\vec{H_2}(w_j^2), u_2}\\
=&0.
\end{array}
\end{equation*}
Similarly, for every $(w_i^1, v_j^2)\in W_{WV}$, the corresponding row in $H$ is $\vec{H}(w_i^1, v_j^2)$. 
Since $u_1\in U_1$, $u_1\in \ker(H_1)$, we know that $\ep{u_1, \vec{H_1}(w_i^1)}=0.$ 
Thus, 
\begin{equation*}
\begin{array}{rl}
& \ep{\vec{H}(w_i^1, v_j^2), (\mathbf{0}, u_1\otimes u_2)}\\
=& 0 + \ep{\vec{H_1}(w_i^1), u_1}
\ep{e_{j}^{V_2}, u_2}\\
=&0.
\end{array}
\end{equation*}

(2). For any row vector $\vec{H_r}$ in $H$, we need to show that $\ep{\vec{H_r}, (u_1^{T}\otimes u_2^{T}), \mathbf{0}}=0$ for any $u_1^{T}\in U_1^{T}$ and $u_2^{T}\in U_2^{T}$.

For every $(v_i^1, w_j^2)\in W_{VW}$, the corresponding row in $H$ is $\vec{H}(v_i^1, w_j^2)$. 
Since $u_1^{T}\in U_1^{T}$, $u_1^{T}\in \ker(H_1^{T})$, we know that $\ep{u_1^{T}, \vec{H_1}^{T}(w_i^1)}=\ep{u_1^{T}, \vec{H_1}(v_i^1)}=0.$ 
Thus, 
\begin{equation*}
\begin{array}{rl}
& \ep{\vec{H}(v_i^1, w_j^2), ( u_1^T\otimes u_2^T, \mathbf{0})}\\
=& \ep{\vec{H_1}(v_i^1), u_1^T}
\ep{e_{j}^{U_2}, u_2^T} + 0\\
=&0.
\end{array}
\end{equation*}
Similarly, for every $(w_i^1, v_j^2)\in W_{WV}$, the corresponding row in $H$ is $\vec{H}(w_i^1, v_j^2)$. 
Since $u_2^{T}\in U_2^{T}$, $u_2^{T}\in \ker(H_2^{T})$, we know that $\ep{u_2^{T}, \vec{H_2}^{T}(w_j^2)}=\ep{u_2^{T}, \vec{H_2}(v_j^2)}=0.$ 
Thus, 
\begin{equation*}
\begin{array}{rl}
& \ep{\vec{H}(w_i^1, v_j^2), (u_1^T\otimes u_2^T), \mathbf{0}}\\
=& \ep{e_{i}^{U_1}, u_1^T}
\ep{\vec{H_2}(v_j^2), u_2^T} + 0\\
=&0.
\end{array}
\end{equation*}
\end{proof}

Note that $\phi(U_1\otimes U_2)\bigcap \phi^{T}(U_1^{T}\otimes U_2^{T})=\{\mathbf{0}\}$,
so $\phi(U_1\otimes U_2)\oplus\phi^T(U_1^T\otimes U_2^T)$ is a direct sum,
and $\phi(U_1\otimes U_2)\oplus\phi^T(U_1^T\otimes U_2^T)\subseteq \ker(H)$.
Moreover, we have the following lemma.
\begin{lem}\label{lem:241221-4}
$\left( 
\phi(U_1\otimes U_2)\oplus\phi^T(U_1^T\otimes U_2^T)
\right) \bigcap \row(H)=\{\mathbf{0}\}$.
\end{lem}
\begin{proof}
Any vector in $\phi(U_1\otimes U_2)\oplus\phi^T(U_1^T\otimes U_2^T)$ can be represented as $(\xi_2, \xi_1)$, where $\xi_1\in U_1\otimes U_2$ and $\xi_2\in U_1^T\otimes U_2^T$.
We claim that $\row(H)\big|_{\BFF_2^{V_{VV}}}\bigcap U_1\otimes U_2=\{\mathbf{0}\}$.

Let $R_V:=\row(H)\big|_{\BFF_2^{V_{VV}}}$. 
We know that $\row(H)$ is spanned by
\begin{equation*}
\left\{\begin{array}{r}
\left(e_{i}^{W_1}\otimes \vec{H_2}(v_j^2), \;\;
\vec{H_1}(w_i^1)\otimes e_{j}^{V_2}\right),
\left(\vec{H_1}(v_k^1)\otimes e_{l}^{W_2}, \;\;
e_{k}^{V_1}\otimes \vec{H_2}(w_l^2)\right)\bigg|
\text{for any }v_k^1\in V_1,v_j^2\in V_2,\\
w_i^1\in W_1, w_l^2\in W_2
\end{array}\right\}.
\end{equation*}
So, $R_V$ is spanned by 
\begin{equation*}
\left\{
\left(\vec{H_1}(w_i^1)\otimes e_{j}^{V_2}\right),
\left(e_{k}^{V_1}\otimes \vec{H_2}(w_l^2)\right)\bigg|
\text{for any }v_k^1\in V_1, w_i^1\in W_1
\right\}.
\end{equation*}
Thus, 
$$R_V=\left(\row(H_1)\otimes\BFF_2^{V_2}\right)+\left(\BFF_2^{V_1}\otimes \row(H_2)\right).$$
Since 
\begin{equation*}
\begin{aligned}
\BFF_2^{V_1}\otimes\BFF_2^{V_2}
=&\left(\row(H_1)\oplus U_1\oplus \Gamma_1\right)\otimes\left(\row(H_2)\oplus U_2\oplus \Gamma_2\right)\\
=&(\row(H_1)\otimes\row(H_2))\oplus\cdots\oplus(\Gamma_1\otimes\Gamma_2)\oplus(U_1\otimes U_2),
\end{aligned}
\end{equation*}
and $R_V\subseteq(\row(H_1)\otimes\row(H_2))\oplus\cdots\oplus(\Gamma_1\otimes\Gamma_2)$, we have that $R_V\bigcap U_1\otimes U_2=\{\mathbf{0}\}$.

Let $R_W:=\row(H)\big|_{\BFF_2^{V_{WW}}}$.
Similarly, we have that $R_W\bigcap U_1^T\otimes U_2^T=\{\mathbf{0}\}$.
Therefore, we have proven this lemma.
\end{proof}

With the lemmas above, we are now able to prove the proposition.
\begin{proof}[Proof of Proposition \ref{prop:lb-rank-H}]
By Lemmas \ref{lem:241221-2} and \ref{lem:241221-4}, we have 
\begin{equation*}
\begin{aligned}
&\dim\ker(H)-\dim\row(H)\\
=&|V|-2\rank(H)\\
\geq&\dim(U_1\otimes U_2)+\dim(U_1^T\otimes U_2^T)\\
=&\left(\dim\ker(H_1)-\dim\row(H_1)\right)\left(\dim\ker(H_2)-\dim\row(H_2)\right)\\
&+\left(\dim\ker(H_1^T)-\dim\row(H_1^T)\right)\left(\dim\ker(H_2^T)-\dim\row(H_2^T)\right)\\
=&\left(|V_1|-2\rank(H_1)\right)\left(|W_1|-2\rank(H_1^T)\right)+
\left(|V_2|-2\rank(H_2)\right)\left(|W_2|-2\rank(H_2^T)\right).
\end{aligned}
\end{equation*}
\end{proof}

Finally let us estimate the minimal distance of $Q\left(\mathcal{G}_1 \boxtimes \mathcal{G}_2\right)$. 
The following lemma is a corollary of Theorem 15 in \cite{TZ13}.

\begin{lem}\label{250101lem2}
For $i \in\{1,2\}$, let $d_i$ be the minimum distance of the classic code with Tanner graph $\mathcal{G}_i$ and let $d_i^T$ denote the minimum distance of the classic code specified by the transpose Tanner graph $\mathcal{G}_i^T$. 
The minimum distance $d_Q$ of the quantum code $\mathrm{Q}\left(\mathcal{G}_1 \boxtimes \mathcal{G}_2\right)$ satisfies
$$
d_Q \geq \min \left(d_1, d_2, d_1^T, d_2^T\right).
$$
\end{lem}

\bibliographystyle{unsrt}
\bibliography{reference}{}

\end{document}